\documentclass[12pt]{amsart}
\usepackage{amssymb}
\usepackage{color}
\usepackage{amsmath,epic,curves,amscd}
\usepackage[english]{babel}
\usepackage{graphicx}
\usepackage{comment}
\usepackage{appendix}
\usepackage{mathdots}
\usepackage[all]{xy}
\usepackage{mathtools}
\usepackage{lscape}
\usepackage{stmaryrd}
\usepackage{mathabx}
\newtheorem{claim}{}[section]
\newtheorem{theorem}[claim]{Theorem}

\newtheorem{lemma}[claim]{Lemma}
\newtheorem{proposition}[claim]{Proposition}
\newtheorem{corollary}[claim]{Corollary}

\theoremstyle{remark}

\renewenvironment{proof}{\noindent{\it Proof. \hskip0pt}}
                      {$\square$\par\medskip}
\allowdisplaybreaks \textwidth 15.5 true cm \textheight 24.2 true cm
\usepackage{color}

\begin{document}
\baselineskip 6.0 truemm
\parindent 1.5 true pc

\newcommand\lan{\langle}
\newcommand\ran{\rangle}
\newcommand\tr{\operatorname{Tr}}
\newcommand\ot{\otimes}
\newcommand\ttt{{\text{\sf t}}}
\newcommand\rank{\ {\text{\rm rank of}}\ }
\newcommand\choi{{\rm C}}
\newcommand\dual{\star}
\newcommand\flip{\star}
\newcommand\cp{{{\mathbb C}{\mathbb P}}}
\newcommand\ccp{{{\mathbb C}{\mathbb C}{\mathbb P}}}
\newcommand\pos{{\mathcal P}}
\newcommand\tcone{T}
\newcommand\mcone{K}
\newcommand\superpos{{{\mathbb S\mathbb P}}}
\newcommand\blockpos{{{\mathcal B\mathcal P}}}
\newcommand\jc{{\text{\rm JC}}}
\newcommand\dec{{\mathbb D}{\mathbb E}{\mathbb C}}
\newcommand\decmat{{\mathcal D}{\mathcal E}{\mathcal C}}
\newcommand\ppt{{\mathcal P}{\mathcal P}{\mathcal T}}
\newcommand\pptmap{{\mathbb P}{\mathbb P}{\mathbb T}}
\newcommand\xxxx{\bigskip\par {\color{red}\sf ========= Stop to read here ========= }}
\newcommand\join{\vee}
\newcommand\meet{\wedge}
\newcommand\ad{{\text{\rm Ad}}}
\newcommand\ldual{\varolessthan}
\newcommand\rdual{\varogreaterthan}
\newcommand{\slmp}{{\mathcal M}^{\text{\rm L}}}
\newcommand{\srmp}{{\mathcal M}^{\text{\rm R}}}
\newcommand{\smp}{{\mathcal M}}
\newcommand{\id}{{\text{\rm id}}}
\newcommand\tsum{\textstyle\sum}
\newcommand\hada{\Theta}
\newcommand\ampl{\mathbb A^{\text{\rm L}}}
\newcommand\ampr{\mathbb A^{\text{\rm R}}}
\newcommand\amp{\mathbb A}
\newcommand\rk{{\text{\rm rank}}\,}
\newcommand\calI{{\mathcal I}}
\newcommand\bfi{{\bf i}}
\newcommand\bfj{{\bf j}}
\newcommand\bfk{{\bf k}}
\newcommand\bfl{{\bf l}}
\newcommand\bfzero{{\bf 0}}
\newcommand\bfone{{\bf 1}}
\newcommand\calc{{\mathcal C}}
\newcommand\calm{{\mathcal M}}
\newcommand\calb{{\mathcal B}}
\newcommand\call{{\mathcal L}}
\newcommand\sr{{\text{\rm SR}}\,}
\newcommand\conv{{\text{\rm conv}}\,}
\newcommand\bil{{\mathcal B}{\mathcal L}}
\newcommand\lin{{\mathcal L}}
\newcommand\pr{\prime}
\newcommand\rana{\rangle^\square}
\newcommand\ranb{\rangle^\diamond}
\newcommand\e\varepsilon
\newcommand\fl{{\text{\sf fl}}}
\newcommand\cals{{\mathcal S}}
\newcommand\invol{\dagger}
\newcommand\pron{{\color{blue}\mathcal P_n^{\sf e}}}
\newcommand\prom{{\color{blue}\mathcal P_m^{\sf e}}}

\title{Choi matrices revisited. III}

\author{Kyung Hoon Han and Seung-Hyeok Kye}
\address{Kyung Hoon Han, Department of Data Science, The University of Suwon, Gyeonggi-do 445-743, Korea}
\email{kyunghoon.han at gmail.com}
\address{Seung-Hyeok Kye, Department of Mathematics and Institute of Mathematics, Seoul National University, Seoul 151-742, Korea}
\email{kye at snu.ac.kr}
%\thanks{Both KHH and SHK were partially supported by NRF-2020R1A2C1A01004587, Korea}

\thanks{partially supported by NRF-2020R1A2C1A01004587, Korea}

\begin{abstract}
We look for all linear isomorphisms from the mapping spaces onto the tensor products of matrices which send
$k$-superpositive maps onto unnormalized bi-partite states of Schmidt numbers less than or equal to $k$.
They also send $k$-positive maps onto $k$-block-positive matrices.
We also look for all the bilinear pairings between the mapping spaces and tensor products of matrices
which retain the usual duality between $k$-positivity and Schmidt numbers $\le k$.
They also retain the duality between $k$-superpositivity and $k$-block-positivity.
\end{abstract}
\maketitle

%%%%%%%%%%%%%%%%%%%%%%%%%%%%%%%%%%%%%%%%%%%%%%%%%%%%%%%%%%%%%%%%%%%%%%%%%%%%%%%%%%%%%%%%%
%%%%%%%%%%%%%%%%%%%%%%%%%%%%%%%%%%%%%%%%%%%%%%%%%%%%%%%%%%%%%%%%%%%%%%%%%%%%%%%%%%%%%%%%%
%%%%%%%%%%%%%%%%%%%%%%%%%%%%%%%%%%%%%%%%%%%%%%%%%%%%%%%%%%%%%%%%%%%%%%%%%%%%%%%%%%%%%%%%%
%%%%%%%%%%%%%%%%%%%%%%%%%%%%%%%%%%%%%%%%%%%%%%%%%%%%%%%%%%%%%%%%%%%%%%%%%%%%%%%%%%%%%%%%%
%%%%%%%%%%%%%%%%%%%%%%%%%%%%%%%%%%%%%%%%%%%%%%%%%%%%%%%%%%%%%%%%%%%%%%%%%%%%%%%%%%%%%%%%%
%%%%%%%%%%%%%%%%%%%%%%%%%%%%%%%%%%%%%%%%%%%%%%%%%%%%%%%%%%%%%%%%%%%%%%%%%%%%%%%%%%%%%%%%%
\section{Introduction}

Choi matrices \cite{choi75-10} and bilinear pairings between linear mapping spaces $\call(M_m,M_n)$
on matrix algebras and tensor products $M_m\ot M_n$ of them have been playing
fundamental roles in current quantum information theory since its beginning,
as we see in the work of Woronowicz \cite{woronowicz}
on the decomposability of positive maps on low dimensional matrix algebras, as well as
Horodecki's separability criterion \cite{horo-1} by positive maps.

Choi matrices give rise to the correspondences
$$
\phi\mapsto \choi_\phi=\sum_{i,j}e_{ij}\ot \phi(e_{ij}): \call(M_m,M_n)\to M_m\ot M_n,
$$
with the usual matrix units $\{e_{ij}\}$, between the convex cones
$\superpos_k$ of all $k$-superpositive maps \cite{ssz} and the
convex cones $\cals_k$ \cite{eom-kye} consisting of unnormalized
bi-partite states whose Schmidt numbers \cite{terhal-sghmidt} are no
greater than $k$. Recall that a state is called separable when it
belongs to $\cals_1$, and entangled otherwise. Recall also that a
linear map is $1$-superpositive \cite{ando-04} if and only if it is
entanglement breaking \cite{{hsrus},{cw-EB}}. The convex cones
$\mathbb P_k$ of all $k$-positive maps \cite{stine} also correspond
to the convex cones $\blockpos_k$ of all $k$-block-positive matrices
\cite{jam_72} through Choi matrices. When $k$ is the minimum of the
sizes of matrices in the domains and the ranges, both results
recover the original work of Choi \cite{choi75-10} for the
correspondence between completely positive maps and positive
(semi-definite) matrices, together with the Kraus decomposition
\cite{kraus-71} of completely positive maps. See
\cite{{kye_comp-ten},{kye_lec_note}} for surveys on the topics. We
summarize in the following diagram:
\begin{equation}\label{diagram}
\xymatrix{
\call(M_m,M_n): & \superpos_1  \ar@{}[r]|-*[@]{\subset}
& \superpos_k \ar@{}[r]|-*[@]{\subset} \ar[d] \ar@{.>}[rrd]^{\rm dual}
& \cp \ar@{}[r]|-*[@]{\subset} &\mathbb P_k \ar@{}[r]|-*[@]{\subset} \ar@{.>}[lld] \ar[d] &\mathbb P_1 \\
M_m\ot M_n: & {\mathcal S}_1  \ar@{}[r]|-*[@]{\subset}
&{\mathcal S}_k \ar@{}[r]|-*[@]{\subset} \ar[u] \ar@{.>}[rru]
&\mathcal P \ar@{}[r]|-*[@]{\subset} &\blockpos_k \ar@{}[r]|-*[@]{\subset} \ar[u] \ar@{.>}[llu] &\blockpos_1
}
\end{equation}

In the previous papers \cite{{kye_Choi_matrix},{han_kye_choi_2}}
under the same title, we have shown that a linear isomorphism from
$\call(M_m,M_n)$ onto $M_m\ot M_n$ is of the form $\phi\mapsto
\sum_k e_k\ot\phi( f_k)$ for bases $\{e_k\}$ and $\{ f_k\}$ of
$M_m$ if and only if it can be written by
$$
\phi\mapsto\choi_{\phi\circ\sigma}=\sum_{i,j}e_{ij}\ot \phi(\sigma(e_{ij}))
$$
for a linear isomorphism $\sigma$ on $M_m$, and this isomorphism retains all the vertical
correspondences in the diagram (\ref{diagram}) if and only if $\sigma=\ad_s$ with a nonsingular $s\in M_m$, where the map $\ad_s$ is defined by
$\ad_s(x)=s^*xs$. See also \cite{Paulsen_Shultz} and \cite{kye_comp-ten} for related preceding works.
The matrices $s$ in these variants of Choi matrix play the roles of separating and cyclic vectors
when we consider infinite dimensional analogues of Choi matrices in the recent work \cite{HKS_inf_Choi}.

The first purpose of this note is to consider all possible linear isomorphisms
from the mapping space $\call(M_m,M_n)$ onto the tensor product $M_m\ot M_n$,
and characterize those retaining all the correspondences in the diagram (\ref{diagram}).
Note that all isomorphisms from $\call(M_m,M_n)$ to $M_m\ot M_n$ are of the form
$$
\phi \mapsto \choi^\Theta_\phi:=\Theta(\choi_\phi)
$$
for linear isomorphisms $\Theta:M_m\ot M_n\to M_m\ot M_n$,
and so our task is to look for $\Theta$ which preserves all the convex cones in the bottom row of (\ref{diagram}).
It turns out that the variant $\choi_{\phi\circ\sigma}$ considered in \cite{han_kye_choi_2} is nothing but $\choi^{\sigma^*\ot\id}_\phi$
with the above notation.

Choi matrices are also important to consider the duality described in the diagram (\ref{diagram}).
For $X=\call(M_m,M_n)$ and $Y=M_m\ot M_n$, we begin with the bilinear pairing
\begin{equation}\label{bilin-dir}
\lan\phi,x\ot y\ran_{X,Y}:=\lan \phi(x),y\ran_{M_n},\qquad \phi\in\call(M_m,M_n),\ x\in M_m,\ y\in M_n,
\end{equation}
which is determined by a bilinear pairing on the range. Then all possible bilinear pairings between $\call(M_m,M_n)$ and $M_m\ot M_n$ are of the form
$$
\lan\phi,z\ran_\Theta=\lan\phi,\Theta^{-1}(z)\ran_{X,Y}, \qquad \phi\in \call(M_m,M_n),\ z\in M_m\ot M_n,
$$
for linear isomorphisms $\Theta:M_m\ot M_n\to M_m\ot M_n$.
The second purpose of this note is to look for $\Theta$ with which the bilinear pairing $\lan\ ,\ \ran_\Theta$ retains all the dualities in (\ref{diagram}).

We show that $\phi\mapsto\choi^\Theta_\phi$ retains the vertical correspondence
between $\superpos_k$ and $\cals_k$  in (\ref{diagram})
if and only if $\lan\ ,\ \ran_\Theta$ retains the duality between $\mathbb P_k$ and $\cals_k$  in (\ref{diagram}) if and only if
$\Theta(\cals_k)=\cals_k$.
When $k=m\meet n$, the minimum of $m$ and $n$, $\cals_{m\meet n}$ is nothing but the convex cone of all
positive matrices. We use the results  on the positivity preserving linear maps \cite{{schneider},{molnar},{semrl_souror}}
together with Schmidt rank $k$ non-decreasing linear maps \cite{chan_lim},
to show that $\Theta(\cals_k)=\cals_k$ for every $k=1,2,\dots, m\meet n$
if and only if
$\Theta$ is one of the following
\begin{equation}\label{exam}
\ad_s\ot \ad_t,\qquad \ttt_m\ot \ttt_n,\qquad \fl\quad {\rm when}\ m=n,
\end{equation}
together with their composition, with nonsingular $s\in M_m$ and $t\in M_n$,
where $\ttt_m$ and $\fl$ denote the transpose map and the flip operation given  by $\fl (a\ot b)=b\ot a$, respectively.
These maps also satisfy $\Theta(\blockpos_k)=\blockpos_k$ for every $k=1,\dots, m\meet n$.

We will begin with linear isomorphisms and bilinear pairings in the level of $*$-vector spaces, which are vector spaces over the complex field
with conjugate linear involutions $x\mapsto x^*$. See \cite{choi-effros, Paulsen_Tomforde}.
For a given $*$-vector space $X$, the set of all Hermitian elements $x\in X$ satisfying $x^*=x$ will be denoted by $X_h$, which is
a real vector space. Recall that every element $x \in X$ can be written uniquely as
\begin{equation}\label{decom}
x = x_1 + {\rm i} x_2, \qquad x_1, x_2 \in X_h.
\end{equation}
In fact, we have $x_1 = \frac 12 (x + x^*)$ and $x_2 = \frac 1{2{\rm i}}(x - x^*)$.

In the next section, we clarify the relations between linear isomorphisms and bilinear pairings in terms of dual cones, and provide general principles
to answer our questions.
After we discuss the mapping spaces and tensor products of $*$-vector spaces in Section 3, we
restrict ourselves to the cases of matrix algebras
in Section 4, to get the above mentioned results.
We finish the paper to discuss Choi matrices and bilinear pairings appearing in the literature.
In the Appendix, we also discuss the problem to characterize isomorphisms $\Theta$ on $M_m\ot M_n$
satisfying $\Theta(\cals_k)=\cals_k$ for a fixed $k$ with $1\le k<m\meet n$.

%%%%%%%%%%%%%%%%%%%%%%%%%%%%%%%%%%%%%%%%%%%%%%%%%%%%%%%%%%%%%%%%%%%%%%%%%%%%%%%%%%%%%%%%%
%%%%%%%%%%%%%%%%%%%%%%%%%%%%%%%%%%%%%%%%%%%%%%%%%%%%%%%%%%%%%%%%%%%%%%%%%%%%%%%%%%%%%%%%%
%%%%%%%%%%%%%%%%%%%%%%%%%%%%%%%%%%%%%%%%%%%%%%%%%%%%%%%%%%%%%%%%%%%%%%%%%%%%%%%%%%%%%%%%%
%%%%%%%%%%%%%%%%%%%%%%%%%%%%%%%%%%%%%%%%%%%%%%%%%%%%%%%%%%%%%%%%%%%%%%%%%%%%%%%%%%%%%%%%%
%%%%%%%%%%%%%%%%%%%%%%%%%%%%%%%%%%%%%%%%%%%%%%%%%%%%%%%%%%%%%%%%%%%%%%%%%%%%%%%%%%%%%%%%%
%%%%%%%%%%%%%%%%%%%%%%%%%%%%%%%%%%%%%%%%%%%%%%%%%%%%%%%%%%%%%%%%%%%%%%%%%%%%%%%%%%%%%%%%%
\section{Linear isomorphisms and Bilinear pairings}

Suppose that $X$ and $Y$ are $*$-vector spaces. We say that a
bilinear pairing $\lan\ ,\ \ran_{X,Y}$ is {\sl Hermiticity
preserving} when $\lan x^*,y^* \ran_{X,Y} = \overline{\lan x,y
\ran}_{X,Y}$ for every $x \in X$ and $y \in Y$. This happens if and
only if its restriction on $X_h \times Y_h$ is real valued by
(\ref{decom}). Therefore, a Hermiticity preserving bilinear pairing
between $*$-vector spaces $X$ and $Y$ gives rise to an $\mathbb R$-bilinear pairing between real vector spaces
$X_h$ and $Y_h$ by restriction. Conversely, an $\mathbb R$-bilinear pairing on $X_h\times Y_h$ extends uniquely to the
Hermiticity preserving bilinear pairing on $X\times Y$ by
$$
\lan x_1+{\rm i}x_2, y_1+{\rm i}y_2\ran=
\lan x_1,y_1\ran+{\rm i}\lan x_1,y_2\ran+{\rm i}\lan x_2,y_1\ran-\lan x_2,y_2\ran,
$$
using the identity (\ref{decom}).
We also note that non-degeneracy of a Hermiticity preserving bilinear pairing on $X \times Y$ is equivalent to
that of its restriction on $X_h \times Y_h$ by (\ref{decom}) again.

Suppose that $X$ and $Y$ are finite dimensional $*$-vector spaces sharing the same dimension, and
$\lan\ ,\ \ran_{X,Y}$ is a non-degenerate Hermiticity preserving bilinear pairing on $X\times Y$.
For a subset $S$ of $X_h$, we define the {\sl dual cone}
$S^\circ$ of $S$ with respect to $\lan\ ,\ \ran_{X,Y}$ by
$$
S^\circ=\{y\in Y_h: \lan x,y\ran_{X,Y} \ge 0\ {\text{\rm for every}}\ x\in S\}.
$$
For a subset $T$ of $Y_h$, the dual cone ${}^\circ T$ is also defined by
$$
{}^\circ T=\{x\in X_h: \lan x,y\ran_{X,Y} \ge 0\ {\text{\rm for every}}\ y\in T\}.
$$
It is well known that ${}^\circ S^\circ={}^\circ(S^\circ)$ is the smallest closed convex cone containing $S$,
and so ${}^\circ S^\circ=S$ when $S$ is a closed convex cone of $X_h$.

In this section, we begin with two non-degenerate Hermiticity preserving bilinear pairings $\lan\ ,\ \ran_{X,Y}$ on $X\times Y$ and $\lan\ ,\ \ran_Y$ on $Y\times Y$.
For a linear isomorphism $\Gamma:X\to Y$, we may consider the bilinear pairing
$$
(x,y)\mapsto \lan \Gamma(x),y\ran_Y,\qquad x\in X, y\in Y
$$
on $X\times Y$. It is easy to see that this bilinear pairing is Hermiticity preserving if and only if $\Gamma$ is {\sl Hermiticity preserving},
that is, $\Gamma(x^*)=\Gamma(x)^*$ for every $x\in X$.
We compare two bilinear pairings  $\lan\ ,\ \ran_{X,Y}$ and $\lan \Gamma(\ ),\ \ran_Y$, in terms of the dual cones.
We begin with the following technical lemma.

\begin{lemma}\label{basic-lem}
Suppose that $\Lambda$ is a linear functional on a vector space $Y$, and $\alpha$ is a scalar-valued function on $Y\setminus\{0\}$.
If $y\mapsto \alpha(y)\Lambda(y)$ extends to a linear functional on $Y$, then $\alpha$ is a constant function.
\end{lemma}

\begin{proof}
Take $y_0\in\ker\Lambda$ and $y_1\notin\ker\Lambda$ in $Y$. Then we have
$$
\alpha(y_1)\Lambda(y_1)=\alpha(y_0)\Lambda(y_0)+\alpha(y_1)\Lambda(y_1)=\alpha(y_0+y_1)\Lambda(y_0+y_1)
=\alpha(y_0+y_1)\Lambda(y_1),
$$
which implies $\alpha(y_1)=\alpha(y_0+y_1)$. Furthermore, we have
$$
k \alpha(y_1)\Lambda(y_1)=\alpha(k y_1)\Lambda(k y_1)=k \alpha(k y_1)\Lambda(y_1),
$$
and $\alpha(k y_1)=\alpha(y_1)$ whenever $k\neq 0$. We fix $y_1\notin\ker\Lambda$ in $Y$. For an arbitrary $y\in Y$, put $k=\frac{\Lambda(y)}{\Lambda(y_1)}$.
Then we have $y-ky_1\in\ker\Lambda$, and so
$$
\alpha(y)=\alpha(ky_1+(y-ky_1))=\alpha(ky_1)=\alpha(y_1),
$$
as it was required.
\end{proof}

\begin{proposition}
Suppose that $X$ and $Y$ are finite dimensional $*$-vector spaces
sharing the same dimension. For non-degenerate Hermiticity preserving bilinear
pairings $\lan\ ,\ \ran_\pi$ and $\lan\ ,\ \ran_\sigma$ on $X\times
Y$, the following are equivalent:
\begin{enumerate}
\item[{\rm (i)}]
$\lan\ ,\ \ran_\pi=\lan\ ,\ \ran_\sigma$ up to multiplication by a
positive number,
\item[{\rm (ii)}]
$C^{\circ_\pi}=C^{\circ_\sigma}$ for every closed convex cone $C$ of
$X_h$,
\item[{\rm (iii)}]
${}^{\circ_\pi}D= {}^{\circ_\sigma}D$ for every closed convex cone $D$ of
$Y_h$.
\end{enumerate}
\end{proposition}

\begin{proof}
The implication (i) $\Rightarrow$ (ii) is obvious. For (ii)
$\Rightarrow$ (iii), we take $C={}^{\circ_\pi}D$. Then, we have $D =
({}^{\circ_\pi}D)^{\circ_\pi} = C^{\circ_\pi} = C^{\circ_\sigma} =
({}^{\circ_\pi}D)^{\circ_\sigma}$, which implies ${}^{\circ_\sigma}D =
{}^{\circ_\pi}D$.

Now, we suppose that (iii) holds, and take nonzero  $y\in Y_h$. We
take the convex cone $D$ of $Y_h$ generated by $y$, that is,
$D=\{\lambda y : \lambda \ge 0\}$. By (iii), we have
$\lan x,y\ran_\pi \ge 0$ if and only if $\lan x,y\ran_\sigma \ge 0$,
which implies
$$
\lan x,y\ran_\pi = 0\ \Longleftrightarrow\ \lan x,y\ran_\sigma = 0,
$$
by applying both $x$ and $-x$. Take $x\in X_h$ and a function $f :
Y_h \setminus \{0\} \to X_h \setminus \{0\}$ satisfying $\lan f(y),
y \ran_\pi  \ne 0$ by non-degeneracy. By the relation $\lan
x-\frac{\lan x,y\ran_\pi}{\lan f(y),y\ran_\pi} f(y),y\ran_\pi=0$, we
have
$$
0=\left\lan x-\frac{\lan x,y\ran_\pi}{\lan f(y),y\ran_\pi} f(y),\
y\right\ran_\sigma =\lan x,y\ran_\sigma - \lan x,y\ran_\pi
\frac{\lan f(y),y\ran_\sigma}{\lan f(y),y\ran_\pi}.
$$
By lemma \ref{basic-lem}, we see that $\frac{\lan
f(y),y\ran_\sigma}{\lan f(y),y\ran_\pi}$ is a positive constant
function on $Y_h\setminus\{0\}$. Let
$$
\lan x,y \ran_\pi= \lambda \lan x,y \ran_\sigma, \qquad x \in X_h,~
y \in Y_h
$$
for $\lambda>0$. By (\ref{decom}), it still holds for general $x \in
X$ and $y \in Y$.
\end{proof}

There are two ways to define bilinear pairings between
$X=\call(M_m,M_n)$ and $Y=M_m\ot M_n$; the one is to define directly
by (\ref{bilin-dir}), the other is to use the bilinear form on $Y$
through isomorphisms $\Gamma:X\to Y$.
$$
\xymatrix{X \times Y \ar[d]_{\Gamma \times {\rm id}_Y} \ar[drr]^{\lan\ ,\ \ran_{X,Y}} && \\
Y \times Y \ar[rr]_{\lan\ ,\ \ran_Y}  & & \mathbb C}
$$

\begin{corollary}\label{id-dual}
Let $X$ and $Y$ be $*$-vector spaces
with a Hermiticity preserving linear isomorphism $\Gamma:X\to Y$.
Suppose that $\lan\ ,\ \ran_{X,Y}$ and $\lan\ ,\ \ran_Y$ are
non-degenerate Hermiticity preserving bilinear pairings on $X\times Y$ and
$Y\times Y$, respectively. Then the following are equivalent:
\begin{enumerate}
\item[{\rm (i)}]
$\lan\ ,\ \ran_{X,Y}=\lan \Gamma(\ ),\ \ran_Y$ up to multiplication
by a positive number,
\item[{\rm (ii)}]
$C^{\circ_{X,Y}}=(\Gamma(C))^{\circ_Y}$ for every closed convex cone
$C$ of $X_h$,
\item[{\rm (iii)}]
$\Gamma({}^{\circ_{X,Y}}D)= {}^{\circ_Y}D$ for every closed convex cone
$D$ of $Y_h$.
\end{enumerate}
\end{corollary}

Now, we fix a linear isomorphism $\Gamma:X\to Y$ and a non-degenerate bilinear pairing $\lan\ ,\ \ran_{X,Y}$ on $X\times Y$.
Then every linear isomorphism from $X$ onto $Y$ is of the form
\begin{equation}\label{another-iso}
\Gamma^\Theta:=\Theta\circ\Gamma:X\to Y,
\end{equation}
for a linear isomorphism $\Theta:Y\to Y$. Furthermore, every non-degenerate bilinear pairing on $X\times Y$ is of the form
\begin{equation}\label{another-bil}
\lan x,y\ran_\Theta=\lan x,\Theta^{-1}(y)\ran_{X,Y},\qquad x\in X,\ y\in Y,
\end{equation}
for a linear isomorphism $\Theta:Y\to Y$. To see this, we take bases $\{e_i\}$ and $\{f_i\}$ satisfying $\lan e_i,f_j\ran_{X,Y}=\delta_{ij}$.
Then every non-degenerate bilinear pairing $\lan\ ,\ \ran_{X,Y}'$ on $X\times Y$ is determined by
$\lan e_i,\tilde f_j\ran_{X,Y}'=\delta_{i,j}$ by taking another basis $\{\tilde f_j\}$ of $Y$ \cite[Proposition II.1]{han_kye_choi_2}. We define the linear isomorphism
$\Theta:Y\to Y$ by $\Theta(f_j)=\tilde f_j$. By $\lan e_i,\Theta(f_j)\ran_{X,Y}'=\delta_{i,j}=\lan e_i,f_j\ran_{X,Y}$, we have
the above relation.

The dual map $\Theta^*:Y\to Y$ is defined by $\lan \Theta(y_1),y_2\ran_Y = \lan y_1,\Theta^*(y_2)\ran_Y$ for $y_1, y_2 \in Y$.

\begin{proposition}\label{rel-bi-forms}
Suppose that a linear isomorphism $\Gamma:X\to Y$ and a non-degenerate bilinear pairing $\lan\ ,\ \ran_{X,Y}$ satisfy the relation
$$
\lan \Gamma(x),y\ran_Y = \lan x,y\ran_{X,Y},\qquad x\in X, y\in Y.
$$
For linear isomorphisms $\Theta_1, \Theta_2$ and $\Theta_3$ on $Y$,
the following are equivalent:
\begin{enumerate}
\item[{\rm (i)}]
$\lan \Gamma^{\Theta_2}(x),y\ran_{\Theta_1} = \lan x,y\ran_{\Theta_3}$ for every $x\in X$ and $y\in Y$,
\item[{\rm (ii)}]
$\Theta_1\circ(\Theta_2^*)^{-1}=\Theta_3$.
\end{enumerate}
\end{proposition}

\begin{proof}
For $x\in X$ and $y\in Y$, we have
$$
\begin{aligned}
\lan \Gamma^{\Theta_2}(x),y\ran_{\Theta_1}
&= \lan\Theta_2(\Gamma(x)),\Theta_1^{-1}(y)\ran_Y\\
&= \lan\Gamma(x),\Theta_2^*(\Theta_1^{-1}(y))\ran_Y\\
&=\lan x,\Theta_2^*(\Theta_1^{-1}(y))\ran_{X,Y}.
\end{aligned}
$$
On the other hand, we have $\lan x,y\ran_{\Theta_3}=\lan x,\Theta_3^{-1}(y)\ran_{X,Y}$. Therefore, we see that (i) holds if and only if $\Theta_2^*\circ\Theta_1^{-1} = \Theta_3^{-1}$
holds if and only if (ii) holds.
\end{proof}

Now, we fix a non-degenerate Hermiticity preserving bilinear pairing $\lan\ ,\
\ran_{X,Y}$ on $X\times Y$, and suppose that $\lan\ ,\ \ran_\pi$ is
another Hermiticity preserving bilinear pairing on $X\times Y$. For a closed
convex cone $C$ of $X_h$, we see that
$C^{\circ_\pi}=C^{\circ_{X,Y}}$ if and only if
$C={}^{\circ_\pi}(C^{\circ_{X,Y}})$ if and only if
${}^{\circ_{X,Y}}(C^{\circ_{X,Y}})={}^{\circ_\pi}(C^{\circ_{X,Y}})$,
that is, the dual of $C$ with respect to $\lan\ ,\ \ran_\pi$ in place of
$\lan\ ,\ \ran_{X,Y}$ does not change if and only if the same is true for $C^{\circ_{X,Y}}$.
When this is the case, we say that
$\lan\ ,\ \ran_\pi$ {\sl retains the duality between $C$ and
$C^{\circ_{X,Y}}$}. It is easily seen that the bilinear pairing $\lan\
,\ \ran_\Theta$ in (\ref{another-bil}) is Hermiticity preserving if and only if
$\Theta:Y\to Y$ is Hermiticity preserving.

We also fix a linear isomorphism $\Gamma:X\to Y$, and suppose that $\Theta$ is an isomorphism on $Y$. We say that the isomorphism $\Gamma^\Theta$ in (\ref{another-iso})
{\sl retains the correspondence between $C\subset X$ and $\Gamma(C)\subset Y$} when $\Gamma^\Theta(C)=\Gamma(C)$.

\begin{proposition}\label{basic}
Suppose that $\Gamma:X\to Y$ is a linear isomorphism and
$\lan\ ,\ \ran_{X,Y}$ is a non-degenerate Hermiticity preserving bilinear pairing on $X\times Y$.
For a Hermiticity preserving linear isomorphism $\Theta:Y\to Y$ and a closed convex cone $C$ of $Y_h$, the following are equivalent:
\begin{enumerate}
\item[{\rm (i)}]
$\Gamma^\Theta$ retains the correspondence between $\Gamma^{-1}(C)$ and $C$,
\item[{\rm (ii)}]
$\lan\ ,\ \ran_\Theta$ retains the duality between ${}^{\circ_{X,Y}}C$ and $C$,
\item[{\rm (iii)}]
$\Theta(C)=C$.
\end{enumerate}
\end{proposition}

\begin{proof}
The equivalence between (i) and (iii) is trivial.
We note that $x\in X_h$ belongs to ${}^{\circ_\Theta}C$ if and only if $\lan x,\Theta^{-1}(y)\ran_{X,Y}\ge 0$ for every $y\in C$ if and only if
$\lan x,y^\pr\ran_{X,Y}\ge 0$ for every $y^\pr\in\Theta^{-1}(C)$ if and only if $x\in {}^{\circ_{X,Y}}(\Theta^{-1}(C))$, and so we have
${}^{\circ_\Theta}C={}^{\circ_{X,Y}}(\Theta^{-1}(C))$. Therefore, we see that
$\lan\ ,\ \ran_\Theta$ retains the duality between ${}^{\circ_{X,Y}}C$ and $C$
if and only if ${}^{\circ_{X,Y}}C = {}^{\circ_\Theta}C$
if and only if ${}^{\circ_{X,Y}}C={}^{\circ_{X,Y}}(\Theta^{-1}(C))$
if and only if $C=\Theta^{-1}(C)$ if and only if
(iii) holds.
\end{proof}

If a closed convex cone $C$ spans $X_h$ and a linear isomorphism $\Theta$ satisfies (iii), then $\Theta$ is necessarily Hermiticity preserving by (\ref{decom}).

Suppose that $\lan\ ,\ \ran_X$ is a Hermiticity preserving bilinear form on $X$, and $\Theta:X\to X$ is a Hermiticity preserving linear map.
For an arbitrary subset $S$ in $X_h$ and $x\in X_h$, we have
$$
\begin{aligned}
x\in{}^{\circ_X}(\Theta^*(S))
&\Longleftrightarrow\ \lan x,\Theta^*(y)\ran_X\ge 0\ {\text{\rm for every}}\ y\in S\\
&\Longleftrightarrow\ \lan\Theta(x),y\ran_X\ge 0\ {\text{\rm for every}}\ y\in S\\
&\Longleftrightarrow\ \Theta(x)\in {}^{\circ_X}S\\
&\Longleftrightarrow\ x\in\Theta^{-1}({}^{\circ_X}S),
\end{aligned}
$$
and so we have
${}^{\circ_X}(\Theta^*(S))=\Theta^{-1}({}^{\circ_X}S)$.
When $S$ is a closed convex cone of $X_h$, we replace $S$ by $S^{\circ_X}$, to
get ${}^{\circ_X}(\Theta^*(S^{\circ_X}))=\Theta^{-1}(S)$, from which we also have
$$
\Theta^*(S^{\circ_X})=\Theta^{-1}(S)^{\circ_X}.
$$
If $\Theta$ is a bijection, then $\Theta(S)=S$ if and only if $S=\Theta^{-1}(S)$, and so we have the following:

\begin{proposition}\label{hvmmj}
Suppose that $\lan\ ,\ \ran_X$ is a Hermiticity preserving bilinear form on $X$,
and $\Theta:X\to X$ is a Hermiticity preserving linear isomorphism.
For a closed convex cone $C$ of $X_h$, we have $\Theta(C)=C$ if and only if $\Theta^*(C^{\circ_X})=C^{\circ_X}$.
\end{proposition}

%%%%%%%%%%%%%%%%%%%%%%%%%%%%%%%%%%%%%%%%%%%%%%%%%%%%%%%%%%%%%%%%%%%%%%%%%%%%%%%%%%%%%%%%%
%%%%%%%%%%%%%%%%%%%%%%%%%%%%%%%%%%%%%%%%%%%%%%%%%%%%%%%%%%%%%%%%%%%%%%%%%%%%%%%%%%%%%%%%%
%%%%%%%%%%%%%%%%%%%%%%%%%%%%%%%%%%%%%%%%%%%%%%%%%%%%%%%%%%%%%%%%%%%%%%%%%%%%%%%%%%%%%%%%%
%%%%%%%%%%%%%%%%%%%%%%%%%%%%%%%%%%%%%%%%%%%%%%%%%%%%%%%%%%%%%%%%%%%%%%%%%%%%%%%%%%%%%%%%%
%%%%%%%%%%%%%%%%%%%%%%%%%%%%%%%%%%%%%%%%%%%%%%%%%%%%%%%%%%%%%%%%%%%%%%%%%%%%%%%%%%%%%%%%%
%%%%%%%%%%%%%%%%%%%%%%%%%%%%%%%%%%%%%%%%%%%%%%%%%%%%%%%%%%%%%%%%%%%%%%%%%%%%%%%%%%%%%%%%%
\section{Hermiticity preserving bilinear pairings on mapping spaces and tensor products}

In this section, we consider  the case when $X$ and $Y$ are the mapping space $\call(V,W)$ and the tensor product $V\ot W$, respectively,
for given finite dimensional vector spaces $V$ and $W$.

We begin with fixed non-degenerate bilinear forms $\lan\ ,\ \ran_V$ and $\lan\ ,\ \ran_W$ on
finite dimensional vector spaces $V$ and $W$, respectively.
They give rise to the bilinear form
\begin{equation}\label{tensor-id}
\lan v_1\ot w_1, v_2 \ot w_2 \ran_{V\ot W}=\lan v_1 ,v_2\ran_V\lan w_1, w_2\ran_W,\qquad v_1, v_2\in V,\ w_1, w_2 \in W,
\end{equation}
on the space $V\ot W$.
We take bases $\{e_i\}$ and $\{f_i\}$ of $V$ satisfying $\lan e_i, f_j\ran_V=\delta_{ij}$.
In this section, we also consider the linear map $\Gamma : X \to Y$, which is defined by
\begin{equation}\label{choi-defi}
\Gamma : \phi\in\call(V,W) \mapsto \choi_\phi := \sum_i e_i\ot\phi(f_i)\in V\ot W.
\end{equation}
Then $\choi_\phi$, which plays the role of Choi matrices,
does not depend  on the choice of bases, but depends only on
the bilinear form $\lan\ ,\ \ran_V$ given by $\lan e_i, f_j\ran_V=\delta_{ij}$ on the domain space \cite{han_kye_choi_2}.

If $V$ and $W$ are $*$-vector spaces, then their tensor product $V \otimes W$  is also a $*$-vector space with the involution given  by
$$
(v \otimes w)^* = v^* \otimes w^*,\qquad v\in V,\ w\in W.
$$
It is known \cite{choi-effros} that
\begin{equation}\label{hermitian}
(V \otimes W)_h = V_h \otimes_{\mathbb R} W_h.
\end{equation}
When $\lan\ ,\ \ran_V$ and $\lan\ ,\ \ran_W$ are Hermiticity preserving, so is
the bilinear form $\lan\ ,\ \ran_{V\ot W}$ in (\ref{tensor-id}).
Conversely, if the bilinear form $\lan\ ,\ \ran_{V\ot W}$ is
Hermiticity preserving, then so are both $\lan\ ,\ \ran_V$ and $\lan\ ,\ \ran_W$
up to complex scalar multiples. To see that $\lan\ ,\ \ran_V$
(respectively $\lan\ ,\ \ran_W$) is Hermiticity preserving up to scalar
multiples, we may fix Hermitian $w_1$ and $w_2$
(respectively $v_1$ and $v_2$) in (\ref{tensor-id}).

When $V$ and $W$ are $*$-vector spaces, the mapping space
$\call(V,W)$ is also a $*$-vector space with the involution
$\phi\mapsto\phi^\invol$ defined by
$$
\phi^\invol(v)=\phi(v^*)^*,\qquad v\in V.
$$
Then, $\phi$ is Hermiticity preserving if and only if it is Hermitian in the $*$-vector space $\call(V,W)$.
In particular, the dual spaces $V^{\rm d}$ and $W^{\rm d}$ are also $*$-vector spaces.
For a non-degenerate bilinear pairing $\lan\ ,\ \ran_{V,W}$ on $V\times W$, we recall that the duality map
$D:V\to W^{\rm d}$  is defined by
$D(v)(w)=\lan v,w\ran_{V,W}$. We have
$$
D(v^*)(w) = \lan v^*,w \ran_{V,W}, \qquad
D(v)^\invol(w)=\overline{D(v)(w^*)}=\overline{\lan v,w^*\ran}_{V,W},
$$
and so we see that a bilinear pairing $\lan\ ,\ \ran_{V,W}$ on $V\times W$ is Hermiticity preserving if and only if its duality map $D:V\to W^{\rm d}$ is Hermiticity preserving.

We call a basis of a $*$-vector space is {\sl Hermitian} if it consists of Hermitian vectors.
Suppose that $\{e_i : i \in I\}$ is a basis of the real vector space $X_h$. By (\ref{decom}), it spans the whole space $X$.
If $\sum_i \alpha_i e_i = 0$ for $\alpha_i \in \mathbb C$, then we have $\sum_i \alpha_i e_i = (\sum_i \alpha_i e_i)^* = \sum_i \bar \alpha_i e_i$,
and so $\sum_i {\rm i}(\alpha_i-\bar \alpha_i) e_i = 0$.
Since each coefficient ${\rm i}(\alpha_i-\bar \alpha_i)$ is real, we have $\alpha_i = \bar \alpha_i$, that is, $\alpha_i \in \mathbb R$, thus $\alpha_i=0$.
Therefore, every basis of $X_h$ itself gives rise to a Hermitian basis of $X$.
The following proposition is a $*$-vector space version of \cite[Proposition II.1]{han_kye_choi_2}.

\begin{proposition}\label{basic_basis}
For a given bilinear pairing $\lan\ ,\ \ran_{X,Y}$ between finite dimensional vector spaces
$X$ and $Y$ with the same dimension, the following are equivalent:
\begin{enumerate}
\item[{\rm (i)}]
$\lan\ ,\ \ran_{X,Y}$ is non-degenerate Hermiticity preserving,
\item[{\rm (ii)}]
there exist a Hermitian basis $\{e_i:i\in I\}$ of $X$ and a Hermitian basis $\{f_i:i\in I\}$ of $Y$  satisfying $\lan e_i,f_j\ran_{X,Y}=\delta_{ij}$,
\item[{\rm (iii)}]
for any given Hermitian basis  $\{e_i:i\in I\}$ of $X$, there exists a unique Hermitian
basis $\{f_i:i\in I\}$ of $Y$ satisfying $\lan e_i,f_j\ran_{X,Y}=\delta_{ij}$.
%By \cite[Proposition II.1]{han_kye_choi_2}, we can take a unique basis
\end{enumerate}
\end{proposition}

\begin{proof}
The only nontrivial direction is (i) $\Rightarrow$ (iii).
Take a Hermitian basis $\{e_i:i\in I\}$ of $X$.
By \cite[Proposition II.1]{han_kye_choi_2}, there exists a unique basis $\{f_i:i\in I\}$ of $Y$ satisfying $\lan e_i,f_j\ran_{X,Y}=\delta_{ij}$.
The Hermiticity of $f_j$ follows from
$$
\lan e_i,f_j^*\ran_{X,Y}=\lan e_i^*,f_j^*\ran_{X,Y}=\overline{\lan e_i,f_j\ran}_{X,Y}=\lan e_i,f_j\ran_{X,Y}.
$$
\end{proof}

The following proposition can be regarded as a generalization of \cite{dePillis} to the cases of $*$-vector spaces,
because the Hermiticity preserving property of $\Gamma$ tells us that $\phi$ is Hermiticity preserving if and only if $\choi_\phi$
is Hermitian.

\begin{proposition}\label{gdfnh}
The isomorphism $\Gamma$ in {\rm (\ref{choi-defi})} is Hermiticity preserving if and only if its associated bilinear form $\lan\ ,\ \ran_V$ is Hermiticity preserving.
\end{proposition}

\begin{proof}
We take two bases $\{e_i\}$ and $\{f_i\}$ of $V$ satisfying $\lan e_i,f_j\ran_V=\delta_{ij}$.
If $\lan\ ,\ \ran_V$ is Hermiticity preserving, then we have $\lan e_i^*,f_j^*\ran_V = \overline{\lan e_i,f_j\ran}_V = \delta_{ij}$, that is,
$\{e_i^*\}$ and $\{f_i^*\}$ also define the same bilinear form $\lan\ ,\ \ran_V$.
Since the Choi matrix is independent of the choice of basis, we have
$$
\Gamma(\phi^\invol)^*  = \left(\sum_i e_i \otimes \phi(f_i^*)^*\right)^*
= \sum_i e_i^* \otimes \phi(f_i^*) = \Gamma(\phi).
$$
For the converse, we consider the linear functionals $\phi_j(v) = \overline{\lan e_j,v^*\ran_V}$.
Then we have
$$
\Gamma(\phi_j^\invol)^* = \left(\sum_i e_i \otimes \phi_j^\invol(f_i)\right)^*
= \sum_i \overline{\lan e_j, f_i \ran}_V e_i^*
= e_j^* = \sum_i \lan e_j^*, f_i \ran_V e_i,
$$
and
$$
\Gamma(\phi_j) = \sum_i e_i \otimes \phi_j (f_i)
= \sum_i \overline{\lan e_j, f_i^* \ran}_V e_i,
$$
which implies
$\lan e_j^*, f_i \ran_V = \overline{\lan e_j, f_i^* \ran}_V$.
Therefore, we have
$\lan v^*, w \ran_V =  \overline{\lan v, w^* \ran}_V$ for every $v\in V$ and $w\in W$.
\end{proof}

We have begun with bilinear forms $\lan\ ,\ \ran_V$ and $\lan\ ,\ \ran_W$ on $V$ and $W$ respectively,
to define the bilinear form $\lan\ ,\ \ran_{V\ot W}$ on $V\ot W$ and the linear isomorphism
$\Gamma:\phi\mapsto \choi_\phi$ from $\call(V,W)$ onto $V\ot W$.
The next step is to look for the bilinear pairing $\lan\ ,\ \ran_{X,Y}$ between $X=\call(V,W)$ and $Y=V\ot W$ satisfying the relation in
Corollary \ref{id-dual} (i). We have
$$
\begin{aligned}
\lan\Gamma(\phi), v \ot w \ran_{V\ot W}
&=\lan\choi_\phi, v \ot w \ran_{V\ot W} \\
&=\lan\textstyle\sum_i e_i\ot\phi(f_i), v\ot w \ran_{V\ot W} \\
&= \textstyle\sum_i \lan e_i,v \ran_V \lan \phi(f_i),w \ran_W \\
&= \lan \phi\left(\textstyle\sum_i \lan e_i, v \ran_V f_i\right), w\ran_W \\
&=\lan\phi(v), w\ran_W
\end{aligned}
$$
for $v\in V$ and $w\in W$, and so it is natural to define
\begin{equation}\label{formu_bil}
\lan\phi, v\ot w\ran_{X,Y}:=\lan\Gamma(\phi),v\ot w\ran_Y=\lan\phi(v),w\ran_W,\qquad \phi\in\call(V,W),\ v\in V,\ w\in W,
\end{equation}
as in (\ref{bilin-dir}).
Therefore, the bilinear pairing $\lan\ ,\ \ran_{X,Y}$ depends only on the bilinear form $\lan\ ,\ \ran_W$ on the range space $W$.
It is trivial to see that $\lan\ ,\ \ran_{X,Y}$ is Hermiticity preserving if and only if $\lan\ ,\ \ran_W$ is Hermiticity preserving.

Finally, we also define the bilinear form on the space $\call(V,W)$ by
\begin{equation}\label{map-id}
\lan\phi,\psi\ran_{\call(V,W)}:=\lan\choi_\phi,\choi_\psi\ran_{V\ot W},\qquad \phi,\psi\in\call(V,W).
\end{equation}
We have
\begin{equation}\label{map-iidd}
\lan\phi,\psi\ran_{\call(V,W)}
=\left\lan \textstyle\sum_i e_i\ot\phi(f_i),\sum_j e_j\ot\psi(f_j)\right\ran_{V\ot W}
=\textstyle\sum_{i,j} \lan e_i,e_j\ran_V\lan\phi(f_i),\psi(f_j)\ran_W
\end{equation}
for  bases $\{e_i\}$ and $\{f_i\}$ of $V$ satisfying $\lan e_i,f_j\ran_V = \delta_{i,j}$.
If both $\lan\ ,\ \ran_V$ and $\lan\ ,\ \ran_W$ are Hermiticity preserving then it is easy to see that
$\lan\ ,\ \ran_{\call(V,W)}$ is also Hermiticity preserving.
Since $\{e_i^*\}$ and $\{f_i^*\}$ also define $\lan\ ,\ \ran_V$, we have
$$
\begin{aligned}
\lan\phi^\invol,\psi^\invol\ran_{\call(V,W)}
&=\textstyle\sum_{i,j}\lan e_i,e_j\ran_V\lan \phi^\invol(f_i),\psi^\invol(f_j)\ran_W\\
&=\textstyle\sum_{i,j} \overline{\lan e_i^*,e_j^*\ran}_V \overline{\lan \phi(f_i^*),\psi(f_j^*)\ran}_W\\
&=\overline{\lan\phi,\psi\ran}_{\call(V,W)},
\end{aligned}
$$
as it was required.

From now on, all the above five bilinear pairings
$$
\lan\ ,\ \ran_{V},\qquad
\lan\ ,\ \ran_{W},\qquad
\lan\ ,\ \ran_{V\ot W},\qquad
\lan\ ,\ \ran_{X,Y},\qquad
\lan\ ,\ \ran_{\call(V,W)}
$$
will be denoted by just $\lan\ ,\ \ran$, which should be distinguished
through the contexts.
Note that the third and the last bilinear pairings are determined by both $\lan\ ,\ \ran_{V}$ and $\lan\ ,\ \ran_{W}$.
On the other hand, the forth one is determined by $\lan\ ,\ \ran_{W}$. All of them are Hermiticity preserving whenever both $\lan\ ,\ \ran_{V}$ and
$\lan\ ,\ \ran_{W}$ are Hermiticity preserving.

Recall that every linear isomorphism from $\call(V,W)$ onto $V\ot W$
is given by $\Gamma^\Theta=\Theta\circ\Gamma$ for a linear isomorphism $\Theta$ on $V\ot W$.
When $\Theta=\sigma\ot\tau$ is a simple tensor, the map $\Gamma^{\sigma\ot\tau}$ can be expressed by
the composition of maps by following identity
$$
(\sigma\ot\tau)(\choi_\phi)=\choi_{\tau\circ\phi\circ\sigma^*}
$$
for $\phi\in\call(V,W)$. See Proposition V.2 in \cite{han_kye_choi_2}.
Therefore, we have the identity
$$
\Gamma^{\sigma\ot\id_W}(\phi)=\choi_{\phi\circ\sigma^*}\in V\ot W,\qquad\phi\in \call(V,W).
$$
Varying isomorphisms $\sigma$ on $V$, they exhaust all possible linear isomorphisms from $\call(V,W)$ onto $V\ot W$
which admit the expressions (\ref{choi-defi}) like Choi matrices,
as it was shown in the second part \cite{han_kye_choi_2}.
More precisely, we have the following:

\begin{proposition}[\cite{han_kye_choi_2}]\label{one-side-var}
Suppose that $\Theta$ is a linear isomorphism on $V\ot W$.
Then the following are equivalent:
\begin{enumerate}
\item[{\rm (i)}]
there exist bases $\{e_i\}$ and $\{f_i\}$ of $V$ such that $\Gamma^{\Theta}(\phi)=\sum_i e_i\ot \phi(f_i)$ for all $\phi \in \mathcal L(V,W)$,
\item[{\rm (ii)}]
$\Theta=\sigma\ot\id_W$ for an isomorphism $\sigma$ on $V$.
\end{enumerate}
\end{proposition}

Now, we suppose that $\Theta_1, \Theta_2$ are linear isomorphisms on $V\ot W$,
and we consider the bilinear pairing
$$
(\phi,z)\mapsto \lan\Gamma^{\Theta_2}(\phi), z\ran_{\Theta_1}
$$
for $\phi\in\call(V,W),\ z\in V\ot W$.
We look for conditions on $\Theta_1,\Theta_2$ with which this bi-linear pairing has the identity as in (\ref{formu_bil}).
For an isomorphism $\tau$ on $W$, we have
\begin{equation}\label{id-tau}
\lan\phi(v),w\ran_\tau
=\lan\phi(v),\tau^{-1}(w)\ran
=\lan\phi,v\ot\tau^{-1}(w)\ran
%=\lan\phi,(\id_V\ot\tau^{-1})(v\ot w)\ran
=\lan\phi, v\ot w\ran_{\id_V\ot\tau}.
\end{equation}
By Proposition \ref{rel-bi-forms}, we have the following:

\begin{proposition}\label{bi-form-simple-id}
Suppose that $\Theta_1$ and $\Theta_2$ are linear isomorphisms on $V\ot W$, and $\tau$ is
a linear isomorphism on $W$.
Then the following are equivalent:
\begin{enumerate}
\item[{\rm (i)}]
$\lan\Gamma^{\Theta_2}(\phi),v\ot w\ran_{\Theta_1} = \lan\phi(v),w\ran_\tau$ for every $v\in V$, $w\in W$ and $\phi\in\call(V,W)$,
\item[{\rm (ii)}]
$\Theta_1\circ(\Theta_2^*)^{-1} = \id_V\ot\tau$.
\end{enumerate}
\end{proposition}

%%%%%%%%%%%%%%%%%%%%%%%%%%%%%%%%%%%%%%%%%%%%%%%%%%%%%%%%%%%%%%%%%%%%%%%%%%%%%%%%%%%%%%%%%
%%%%%%%%%%%%%%%%%%%%%%%%%%%%%%%%%%%%%%%%%%%%%%%%%%%%%%%%%%%%%%%%%%%%%%%%%%%%%%%%%%%%%%%%%
%%%%%%%%%%%%%%%%%%%%%%%%%%%%%%%%%%%%%%%%%%%%%%%%%%%%%%%%%%%%%%%%%%%%%%%%%%%%%%%%%%%%%%%%%
%%%%%%%%%%%%%%%%%%%%%%%%%%%%%%%%%%%%%%%%%%%%%%%%%%%%%%%%%%%%%%%%%%%%%%%%%%%%%%%%%%%%%%%%%
%%%%%%%%%%%%%%%%%%%%%%%%%%%%%%%%%%%%%%%%%%%%%%%%%%%%%%%%%%%%%%%%%%%%%%%%%%%%%%%%%%%%%%%%%
%%%%%%%%%%%%%%%%%%%%%%%%%%%%%%%%%%%%%%%%%%%%%%%%%%%%%%%%%%%%%%%%%%%%%%%%%%%%%%%%%%%%%%%%%
\section{$k$-superpositivity, Schmidt number $k$ and $k$-positivity}

For the algebra $M_m$ of $m\times m$ matrices, we fix the basis $\{e^m_{ij}:i,j=,1\dots,m\}$ with the usual matrix units,
which gives rise to the usual Choi matrix
$$
\choi_\phi=\sum_{i,j}e_{ij}\ot\phi(e_{ij})\in M_m\ot M_n
$$
for $\phi\in\call(M_m,M_n)$, and the non-degenerate Hermiticity preserving bilinear form on $M_m$ by
\begin{equation}\label{pairing}
\lan x,y\ran=\sum_{i.j=1}^m x_{ij}y_{ij}=\tr(xy^\ttt),\qquad x=[x_{ij}],\ y=[y_{ij}].
\end{equation}
Another basis may give rise to the same bilinear form. For example,
the $2 \times 2$ Weyl basis consisting of
$$
E_1=\textstyle\frac 1{\sqrt 2}\left(\begin{matrix}1&0\\0&1\end{matrix}\right),\
E_2=\textstyle\frac 1{\sqrt 2}\left(\begin{matrix}1&0\\0&-1\end{matrix}\right),\
E_3=\textstyle\frac 1{\sqrt 2}\left(\begin{matrix}0&1\\1&0\end{matrix}\right),\
E_4=\textstyle\frac 1{\sqrt 2}\left(\begin{matrix}0&-1\\1&0\end{matrix}\right)
$$
gives rise to the same bilinear form, consequently the same Choi matrix.

On the other hand, two bases $\{e^m_{ij}\}$ and $\{e^m_{ji}\}$
give rise to the non-degenerate Hermiticity preserving bilinear form
$$
\lan x,y\ran_\ttt=\sum_{i.j=1}^m x_{ij}y_{ji}=\tr(xy),\qquad x=[x_{ij}],\ y=[y_{ij}].
$$
It can be also obtained by single basis consisting of
Pauli matrices
$$
E_1=\textstyle\frac 1{\sqrt 2}\left(\begin{matrix}1&0\\0&1\end{matrix}\right),\
E_2=\textstyle\frac 1{\sqrt 2}\left(\begin{matrix}1&0\\0&-1\end{matrix}\right),\
E_3=\textstyle\frac 1{\sqrt 2}\left(\begin{matrix}0&1\\1&0\end{matrix}\right),\
E_4=\textstyle\frac 1{\sqrt 2}\left(\begin{matrix}0&-{\rm i}\\{\rm i}&0\end{matrix}\right).
$$
The Hermiticity of $2 \times 2$ Weyl basis and Pauli basis illustrates Proposition \ref{basic_basis}.

We also take the basis $\{e^n_{k,\ell}\}$ of $M_n$ to get the bilinear form on $M_n$.
Following (\ref{choi-defi}) and (\ref{formu_bil}),  we have the linear isomorphism $\Gamma : \call(M_m,M_n) \to M_m\ot M_n$ and
the non-degenerate Hermiticity preserving bilinear pairing on $\call(M_m,M_n)\times (M_m\ot M_n)$ by
\begin{equation}\label{stan-bil}
\Gamma(\phi)=\choi_\phi, \qquad \lan\phi,x\ot y\ran:=\lan\Gamma(\phi),x\ot y\ran=\lan\phi(x),y\ran,
\end{equation}
for $\phi\in\call(M_m,M_n)$, $x\in M_m$ and $y\in M_n$, which has been used in \cite{{stormer-dual},{eom-kye}}. See also \cite[Definition 4.2.1]{stormer_book}.
See \cite{{stormer-dual},{stormer08}} for infinite dimensional cases.

Suppose that $\Theta:M_m\ot M_n\to M_m\ot M_n$ is a Hermiticity preserving linear isomorphism. Then we have
a Hermiticity preserving linear isomorphism
$$
\Gamma^\Theta : \phi \in \call(M_m,M_n) \mapsto \choi^\Theta_\phi:=\Theta(\choi_\phi) \in M_m \otimes M_n,
$$
together with a Hermiticity preserving bilinear pairing
$$
\lan\phi,z\ran_\Theta:=\lan \phi,\Theta^{-1}(z)\ran, \qquad \phi\in\call(M_m,M_n),\ z\in M_m\ot M_n,
$$
and every Hermiticity preserving linear isomorphism and non-degenerate Hermiticity preserving bilinear pairing arise in these ways.

When $\Theta=\id\ot\id$, we see that $\choi^{\id\ot\id}_\phi$ is the usual Choi matrix. We also have
$$
\choi^{\ttt\ot\id}_\phi
=(\choi_\phi)^{\ttt\ot\id}
=\sum_{i,j}e_{ji}\ot\phi(e_{ij}),
$$
which was defined by de Phillis \cite{dePillis} prior to Choi matrix, and used by
Jamio\l kowski \cite{jam_72} to get correspondence between $\mathbb P_1$ and $\blockpos_1$.
We also have
$$
\choi^{\id\ot\ttt}_\phi=(\choi_\phi)^{\id\ot\ttt}
=\sum_{i,j}e_{ij}\ot \phi(e_{ij})^\ttt,
\qquad
\choi_\phi^{\ttt\ot\ttt}
=(\choi_\phi)^{\ttt\ot\ttt}
=\sum_{i,j}e_{ij}\ot\phi(e_{ji})^\ttt.
$$
We note that there exists no bases $\{e_k\}$ and $\{ f_k\}$ of $M_m$ satisfying the expression $\choi^{\id\ot\ttt}_\phi=\sum_k e_k\ot\phi(f_k)$
by Proposition \ref{one-side-var}. The same is true for $\choi^{\ttt\ot\ttt}_\phi$.

We apply Proposition \ref{basic} to see that $\phi\mapsto \choi^\Theta_\phi$ retains the correspondence between  $\superpos_k$ and $\cals_k$
if and only if the bilinear pairing $\lan\ ,\ \ran_\Theta$ retains the duality between $\mathbb P_k$ and $\cals_k$ if and only if $\Theta(\cals_k)=\cals_k$.
When $k=m\meet n$, we note that $\cals_{m\meet n}$ is nothing but the convex cone of all positive matrices, and we have the following:

\begin{theorem}[\cite{{schneider},{molnar},{semrl_souror}}]\label{pre_pos}
A linear isomorphism $\Theta:M_m\ot M_n\to M_m\ot M_n$ satisfies $\Theta(\cals_{m\meet n})=\cals_{m\meet n}$ if and only if
$\Theta=\ad_V$ for a nonsingular $V\in M_m\ot M_n$, $\Theta=\ttt_m\ot \ttt_n$, or their composition.
\end{theorem}

\begin{proposition}\label{preserve}
For a linear isomorphism $\Theta:M_m\ot M_n \to M_m\ot M_n$, the following are equivalent;
\begin{enumerate}
\item[{\rm (i)}]
$\Theta(\cals_k)=\cals_k$ for every $k=1,2,\dots, m\meet n$,
%\item[{\rm (ii)}]
%$\Theta(\cals_1)=\cals_1$ and $\Theta(\cals_k)=\cals_k$ for some $k>1$,
\item[{\rm (ii)}]
$\Theta(\cals_{m \meet n})=\cals_{m \meet n}$ and $\Theta(\cals_k)=\cals_k$ for some $k<{m \meet n}$,
\item[{\rm (iii)}]
$\Theta$ is one of the maps listed in {\rm (\ref{exam})} for nonsingular $s\in M_m$, $t\in M_n$, or their composition.
\end{enumerate}
\end{proposition}

\begin{proof}
The direction (iii) $\Longrightarrow$ (i) is clear, and it remains to prove the direction (ii) $\Longrightarrow$ (iii).
Suppose that
$\Theta(\cals_{m\meet n})=\cals_{m\meet n}$ and $\Theta(\cals_k)=\cals_k$ for a fixed $k$ with $1\le k<m\meet n$.
By Theorem \ref{pre_pos}, we have $\Theta=\ad_V$ or $\Theta=\ad_V\circ(\ttt_m\ot \ttt_n)$ for a nonsingular $V\in M_m\ot M_n$.
Because $\Theta$ is an affine isomorphism on the convex cone $\cals_k$,
we see that $\Theta$ sends an extreme ray of $\cals_k$ to an extreme ray of $\cals_k$. Recall that $\varrho\in\cals_k$
generates an extreme ray if and only if $\varrho=|\zeta\ran\lan\zeta|$ for $|\zeta\ran\in\mathbb C^m\ot\mathbb C^n$
with $\sr|\zeta\ran\le k$, where $\sr|\zeta\ran$ denotes the Schmidt rank of $|\zeta\ran\in\mathbb C^m\ot\mathbb C^n$.
We first consider the case of $\Theta=\ad_V$. In this case, $V$ is a linear isomorphism from $\mathbb C^m\ot\mathbb C^n$ onto itself,
and $\sr(V^*|\zeta\ran)\le k$ whenever $\sr|\zeta\ran\le k$.
By \cite{chan_lim}, we have either
$$
V^*(|\xi\ran\ot |\eta\ran)
=(s \ot t)(|\xi\ran\ot |\eta\ran),
$$
or
$$
V^*(|\xi\ran\ot |\eta\ran)=(s\ot  t)(|\eta\ran\ot |\xi\ran)\quad {\text{\rm with}}\ m=n,
$$
for nonsingular $s\in M_m$ and $t\in M_n$.
In the first case, $\Theta=\ad_V=\ad_{s^*}\ot\ad_{t^*}$.
In the second case, we have $V^*=(s\ot t)\choi_\ttt\in M_n\ot M_n$, where we note that the Choi matrix $\choi_\ttt\in M_n\ot M_n$
of the transpose map $\ttt$ is the matrix representing the flip operator between $\mathbb C^n\ot\mathbb C^n$.
We also note that $\choi_\ttt( x\ot y )\choi_\ttt=y\ot x$ for $x,y\in M_n$.
Therefore, we have $\Theta=(\ad_{s^*}\ot\ad_{t^*})\circ\fl$.

It remains to consider the case of $\Theta=\ad_V\circ(\ttt_m\ot \ttt_n)$.
In this case, we apply the above to $\Theta\circ(\ttt_m\ot \ttt_n)=\ad_V$.
This shows that $\Theta$ is one of the maps listed in (\ref{exam}) or their composition.
\end{proof}

Again, we suppose that $\Theta : M_m\ot M_n\to M_m\ot M_n$ is a  Hermiticity preserving linear isomorphism.
We recall that $\blockpos_k^\circ=\cals_k$ with respect to the bilinear form on $M_m\ot M_n$
given by (\ref{tensor-id}) and (\ref{pairing}). Therefore, we may apply Proposition \ref{hvmmj} to see that
$\Theta(\blockpos_k)=\blockpos_k$ if and only if $\Theta^*(\cals_k)=\cals_k$, and so
if $\Theta(\blockpos_k)=\blockpos_k$ then
$\Theta^*$ must be one of the maps listed in (\ref{exam}) or their composition.
We note that $(\ad_s\ot\ad_t)^*=\ad_{s^\ttt}\ot \ad_{t^\ttt}$ and the maps $\ttt_m\ot \ttt_n$ and $\fl$ are invariant
under taking the dual $\Theta\mapsto \Theta^*$. Applying Proposition \ref{basic}, we have the following:

\begin{theorem}\label{thm}
For a Hermiticity preserving linear isomorphism  $\Theta : M_m\ot M_n\to M_m\ot M_n$,
the following are equivalent:
\begin{enumerate}
\item[{\rm (i)}]
$\Theta$ is one of the maps listed in {\rm (\ref{exam})} for nonsingular $s\in M_m$, $t\in M_n$, or their composition,
\item[{\rm (ii)}]
$\Gamma^\Theta : \phi \in \call(M_m,M_n) \mapsto \choi^\Theta_\phi \in M_m \otimes M_n$
retains the correspondence between $\superpos_k$ and $\cals_k$ for $k=1,2,\dots, m\meet n$,
\item[{\rm (iii)}]
$\Gamma^\Theta : \phi \in \call(M_m,M_n) \mapsto \choi^\Theta_\phi \in M_m \otimes M_n$
retains the correspondence between $\mathbb P_k$ and $\blockpos_k$ for $k=1,2,\dots, m\meet n$,
\item[{\rm (iv)}]
$\lan\ ,\ \ran_\Theta$ retains the duality between $\mathbb P_k$ and $\cals_k$ for  $k=1,2,\dots, m\meet n$,
\item[{\rm (v)}]
$\lan\ ,\ \ran_\Theta$ retains the duality between $\superpos_k$ and $\blockpos_k$ for  $k=1,2,\dots, m\meet n$.
\end{enumerate}
\end{theorem}

We note that the condition (ii) or (iii) of Theorem \ref{thm} already implies that $\Theta$ is Hermiticity preserving. To see this,
we recall that $(M_m\ot M_n)_h=(M_m)_h\ot_{\mathbb R}(M_n)_h$ and $(M_m)_h = M_m^+ - M_m^+$, and so
$\cals_1$ span the real space $(M_m\ot M_n)_h$. Therefore, it follows that if $\Gamma^\Theta$ sends $\cals_1$ to Hermitian elements then both $\Gamma^\Theta$ and $\Theta$
should be Hermiticity preserving.
By our definition of \lq retain the duality\rq,  we note that the statement (iv) tells us
$\mathbb P_k^{\circ_\Theta}=\cals_k$ and $\mathbb P_k={}^{\circ_\Theta}\cals_k$, which mean
$$
\varrho\in\cals_k\ \Longleftrightarrow\ \lan\phi,\varrho\ran_\Theta\ge 0\ {\text{\rm for every}}\ \phi\in\mathbb P_k,
$$
and
$$
\phi\in\mathbb P_k\ \Longleftrightarrow\ \lan\phi,\varrho\ran_\Theta\ge 0\ {\text{\rm for every}}\ \varrho\in\cals_k,
$$
respectively. The same is true for $\superpos_k$ and $\blockpos_k$ in the statement (v).

So far, we found all linear isomorphisms from $\call(M_m,M_n)$ onto
$M_m\ot M_n$ which retain the correspondences in the diagram
(\ref{diagram}), and all non-degenerate Hermiticity preserving bilinear pairings
which retain the dualities in (\ref{diagram}).
In the remainder of this section, we consider the duality between $\cals_k$ and
$\blockpos_k$ on the bottom row as well as the duality between
$\superpos_k$ and $\mathbb P_k$ on the top row in the diagram
(\ref{diagram}).

Recall that every Hermiticity preserving bilinear form on $M_m\ot M_n$ is given by
$$
\lan\varrho_1,\varrho_2\ran_\Theta=\lan\varrho_1,\Theta^{-1}(\varrho_2)\ran,\qquad \varrho_1,\varrho_2\in M_m\ot M_n,
$$
for a Hermiticity preserving linear isomorphism $\Theta$ on $M_m \otimes M_n$.
We apply Proposition \ref{basic} with $X=Y=M_m \ot M_n$ to see that $\lan\ ,\ \ran_\Theta$ retains the duality between $\blockpos_k$ and $\cals_k$
if and only if $\Theta(\cals_k)=\cals_k$, and such $\Theta$'s are found in Proposition \ref{preserve}.

We also see that every non-degenerate Hermiticity preserving bilinear form on $\call(M_m,M_n)$ is of the form
$$
\lan\phi,\psi\ran_\Theta:=\lan\choi_\phi,\choi_\psi\ran_\Theta,\qquad \phi,\psi\in\call(M_m,M_n),
$$
for a Hermiticity preserving linear isomorphism $\Theta$ on $M_m \otimes M_n$.
Thus, $\lan\phi,\psi\ran_\Theta$ retains the duality between $\mathbb P_k$ and $\superpos_k$
if and only if $\Theta(\cals_k)=\cals_k$ if and only if $\Theta$ is again one of (\ref{exam}) or their composition.

%%%%%%%%%%%%%%%%%%%%%%%%%%%%%%%%%%%%%%%%%%%%%%%%%%%%%%%%%%%%%%%%%%%%%%%%%%%%%%%%%%%%%%%%%
%%%%%%%%%%%%%%%%%%%%%%%%%%%%%%%%%%%%%%%%%%%%%%%%%%%%%%%%%%%%%%%%%%%%%%%%%%%%%%%%%%%%%%%%%
%%%%%%%%%%%%%%%%%%%%%%%%%%%%%%%%%%%%%%%%%%%%%%%%%%%%%%%%%%%%%%%%%%%%%%%%%%%%%%%%%%%%%%%%%
%%%%%%%%%%%%%%%%%%%%%%%%%%%%%%%%%%%%%%%%%%%%%%%%%%%%%%%%%%%%%%%%%%%%%%%%%%%%%%%%%%%%%%%%%
%%%%%%%%%%%%%%%%%%%%%%%%%%%%%%%%%%%%%%%%%%%%%%%%%%%%%%%%%%%%%%%%%%%%%%%%%%%%%%%%%%%%%%%%%
%%%%%%%%%%%%%%%%%%%%%%%%%%%%%%%%%%%%%%%%%%%%%%%%%%%%%%%%%%%%%%%%%%%%%%%%%%%%%%%%%%%%%%%%%
\section{Conclusion}

In this paper, we found all isomorphisms from $\call(M_m,M_n)$ onto $M_m\ot M_n$ which retain the correspondences in the diagram (\ref{diagram}),
and found all bilinear pairings between $\call(M_m,M_n)$ and $M_m\ot M_n$ which retain the dualities in (\ref{diagram}).

We recall that there is a natural isomorphism $\call(V,W)=V^{\rm d}\ot W$, by which $f\ot w$ in $V^{\rm d}\ot W$ corresponds to
the map $v\mapsto f(v)w$ in $\call(V,W)$. Therefore, one may expect that the simplest way to get
an isomorphism from $\call(V,W)$ onto $V\ot W$ is to use a duality map $V\to V^{\rm d}$ which depends on a bilinear form on $V$.
Actually, it was shown in the previous paper \cite{han_kye_choi_2} that this is the case when and only when the isomorphism can be expressed
by a formula which looks like a Choi matrix. In this sense, we can say that all the variants of Choi matrices are determined by
bilinear forms on the domain space. In this current paper, we have considered all the possible isomorphisms
from $\call(V,W)$ onto $V\ot W$ beyond them.

As for Choi matrices, $\choi_\phi$ and $\choi^{\ttt\ot\id}_\phi$ are
used in the literature, as they were defined by Choi
\cite{choi75-10} and de Pillis \cite{dePillis}, respectively. By
Theorem \ref{thm}, we see that
$$
\phi\mapsto \choi^{\ttt\ot\id}_\phi
$$
does not retain all the correspondences in (\ref{diagram}).
Nevertheless, it retains the correspondence between $\superpos_1$
and $\cals_1$, as well as that between $\mathbb P_1$ and
$\blockpos_1$ by Propositions \ref{basic} and \ref{hvmmj}. In fact,
the isomorphism $\phi\mapsto\choi^{\ttt\ot\id}_\phi$ retains the
correspondence between $\superpos_k$ and $\cals_k$ only when $k=1$.
See
\cite{{Jiang_Luo_Fu_2013},{maj_ty_2013a},{johnson_viola},{Frembs_Cavalcanti}}
for the discussions on the two isomorphisms $\phi\mapsto \choi_\phi$
and $\phi\mapsto\choi^{\ttt\ot\id}_\phi$. When $m=n$, the flip
$\choi^\fl_\phi$ is also defined in \cite[Section 11.3]{beng_zyc}.
The recent paper \cite{schmidt} also discusses another variants of
Choi matrices, which turn out to be $\choi_\phi^\Theta$ in our notation, with
$\Theta=\ad_U$ for a global unitary $U$. Note that $\phi\mapsto \choi_\phi^{\ad_U}$ retains
the correspondence between $\mathbb C\mathbb P$ and
$\mathcal P$ in the diagram (\ref{diagram}) with $k=m\meet n$.
But, this retains the correspondence between $\superpos_k$ and $\cals_k$ for $k<m\meet n$ only when
$U$ is a local unitary.

By the natural isomorphism $\call(V,W)=(V\ot W^{\rm d})^{\rm d}$, we have a natural bilinear pairing between
$\call(V,W)$ and $V\ot W^{\rm d}$. Therefore, the simplest way to get a bilinear pairing between $\call(V,W)$ and $V\ot W$ must be
to define a duality map $W\to W^{\rm d}$, which depends on a bilinear form on the range space $W$.
Recall that the bilinear pairings in (\ref{bilin-dir}) and (\ref{formu_bil}) are determined by a bilinear form on the range space.

We close this paper by examining several bilinear pairings in the literature.
We first examine what happens when the bilinear pairing between $\call(M_m,M_n)$ and $M_m\ot M_n$ is given by a bilinear form on $M_n$,
as in \cite{stormer-dual} and \cite{eom-kye}.
If the bilinear pairing is given by
$$
(\phi,x\ot y)\mapsto \lan \phi(x),y\ran_\tau
$$
for a Hermiticity preserving linear isomorphism $\tau:M_n\to M_n$, then
this pairing retains all the dualities if and only if $\id\ot\tau$ is one of
{\rm (\ref{exam})} or their composition, by (\ref{id-tau}) and Theorem \ref{thm}.
This is the case if and only if $\tau=\ad_t$ for a nonsingular $t\in M_n$.
Suppose that $\Theta_1$ and $\Theta_2$ are linear isomorphisms on $M_m\ot M_n$.
Then we see that the bilinear pairing
$$
(\phi,z)\mapsto \lan \choi^{\Theta_2}_\phi, z\ran_{\Theta_1}
$$
retains all the dualities if and only if $\Theta_1\circ(\Theta_2^*)^{-1}$ is one of
{\rm (\ref{exam})} or their composition, by Proposition \ref{rel-bi-forms} and Theorem \ref{thm}.

A couple of other bilinear pairings have been used in the literature. The bilinear pairing
$$
(\phi,z)\mapsto \lan \choi_\phi, z\ran_{\ttt\ot\ttt},
$$
used in \cite{woronowicz}, retains all the dualities in the diagram
(\ref{diagram}). But, there exists no linear isomorphism
$\tau:M_n\to M_n$ satisfying $\lan \choi_\phi, x\ot
y\ran_{\ttt\ot\ttt}=\lan\phi(x),y\ran_\tau$ by Proposition
\ref{bi-form-simple-id}. In other words, this bilinear pairing is
not determined by a bilinear form on the range space.

On the other hand, the bilinear pairing
$$
(\phi,z)\mapsto \lan \choi_\phi^{\ttt\ot\id}, z\ran_{\ttt\ot\ttt}
$$
was used in \cite{horo-1}.
In this case, we have the relation
$$
\lan \choi_\phi^{\ttt\ot\id}, x\ot y\ran_{\ttt\ot\ttt}=\lan\phi(x),y\ran_\ttt
$$
by $(\ttt\ot\ttt)\circ((\ttt\ot\id)^*)^{-1} = (\id\ot\ttt)$ and Proposition \ref{bi-form-simple-id},
and this bilinear pairing is determined by the bilinear form $\lan\ ,\ \ran_\ttt$ on the range.
Since this bilinear pairing coincides with $\lan \phi, x \otimes y \ran_{\id\ot\ttt}$ by (\ref{id-tau}),
it does not retain all the dualities in the diagram (\ref{diagram}). Especially, the dual cone of $\pos$
is not $\cp$, but the convex cone of all completely copositive maps. Nevertheless,
it retains the duality between $\mathbb P_1$ and $\cals_1$,
as well as that between $\superpos_1$ and $\blockpos_1$ by Propositions \ref{basic} and \ref{hvmmj}.

As for bilinear forms on $\call(M_m,M_n)$, the bilinear form
$$
(\phi,\psi)\mapsto\lan\choi_\phi,\choi_\psi\ran_{\ttt\ot\ttt},\qquad \phi,\psi\in\call(M_m,M_n)
$$
was defined in \cite{ssz}.  See also \cite[Section 11.2]{beng_zyc}.
Since $(\ttt\ot\ttt)(\cals_k)=\cals_k$, we see that this bilinear form retains the duality between $\mathbb P_k$ and $\superpos_k$.
On the other hand, the bilinear form
$$
(\phi,\psi)\mapsto\lan\choi_\phi,\choi_\psi\ran,\qquad \phi,\psi\in\call(M_m,M_n)
$$
has been used in \cite{{gks},{kye_comp-ten},{kye_lec_note}}.

%%%%%%%%%%%%%%%%%%%%%%%%%%%%%%%%%%%%%%%%%%%%%%%%%%%%%%%%%%%%%%%%%%%%%%%%%%%%%%%%%%%%%%%%%
%%%%%%%%%%%%%%%%%%%%%%%%%%%%%%%%%%%%%%%%%%%%%%%%%%%%%%%%%%%%%%%%%%%%%%%%%%%%%%%%%%%%%%%%%
%%%%%%%%%%%%%%%%%%%%%%%%%%%%%%%%%%%%%%%%%%%%%%%%%%%%%%%%%%%%%%%%%%%%%%%%%%%%%%%%%%%%%%%%%
%%%%%%%%%%%%%%%%%%%%%%%%%%%%%%%%%%%%%%%%%%%%%%%%%%%%%%%%%%%%%%%%%%%%%%%%%%%%%%%%%%%%%%%%%
%%%%%%%%%%%%%%%%%%%%%%%%%%%%%%%%%%%%%%%%%%%%%%%%%%%%%%%%%%%%%%%%%%%%%%%%%%%%%%%%%%%%%%%%%
%%%%%%%%%%%%%%%%%%%%%%%%%%%%%%%%%%%%%%%%%%%%%%%%%%%%%%%%%%%%%%%%%%%%%%%%%%%%%%%%%%%%%%%%%
\section{APPENDIX: Isomorphisms preserving separability}

In this paper, we characterized linear isomorphisms $\Theta$ on $M_m\ot M_n$ satisfying $\Theta(\cals_k)=\cals_k$
for every $k=1,2,\dotsm m\meet n$. It seems to be also interesting problems to look for $\Theta$ satisfying $\Theta(\cals_k)=\cals_k$
for a fixed $k$. When $k=m\meet n$, the answer is given in \cite{{schneider},{molnar},{semrl_souror}}, as it is restated in Theorem \ref{pre_pos}. It was also shown in
\cite{{alfsen},{flps}} that $\Theta$ preserves both $\cals_1$ and the trace if and only if
$\Theta$ is one of
\begin{equation}\label{exam_tr}
\ad_s\ot \ad_t,\qquad \ttt_m\ot \id_n, \qquad \id_m\ot\ttt_n,\qquad \fl\quad {\rm when}\ m=n,
\end{equation}
or their composition, with unitaries $s\in M_m$ and $t\in M_n$.
The purpose of this appendix is to show the following:

\begin{theorem}\label{s_1}
A linear isomorphism $\Theta:M_m\ot M_n\to M_m\ot M_n$ satisfies $\Theta(\cals_1)=\cals_1$ if and only if
$\Theta$ is one of {\rm (\ref{exam_tr})} with nonsingular $s\in M_m$ and $t\in M_n$, or their composition.
\end{theorem}

In the first part, we follow the strategy of \cite[Theroem 3]{flps} with a modification.
The absence of trace preserving condition requires further argument involving separation of variables, which will be done in the second part.

We denote by $\mathcal E_n$ the set of all rank one positive operators  on $\mathbb C^n$, which generate all the extreme rays of the cone of $n \times n$ positive matrices.
Since $S^*|\xi \ran \lan \xi| S = |S^*\xi \ran \lan S^*\xi|$, the linear map $\ad_S : M_m \to M_n$ for $S \in M_{m,n}$ maps $\mathcal E_m$ to
$\mathcal E_n$ if and only if $S^*$ is injective if and only if $S$ is surjective.
In this case, $m \le n$ holds necessarily.
By the similar argument as in \cite[Lemma 4]{flps}, we get the following unnormalized version.

\begin{lemma}\label{rank1preserving}
Suppose that $\psi : M_m \to M_n$ is a Hermiticity preserving linear map and satisfies $\psi(\mathcal E_m) \subset \mathcal E_n$. Then one of the following holds;
\begin{enumerate}
\item[(i)] there exist $R \in \mathcal E_n$ and a faithful positive functional $f$ on $M_m$  such that $\psi(A)= f(A)R$,
\item[(ii)] $m \le n$ and there is a surjective $S \in M_{m,n}$ such that $\psi$ has the form
$$
\psi(A)=S^*AS \qquad \text{or} \qquad \psi(A)=S^*A^\ttt S.
$$
\end{enumerate}
\end{lemma}

In order to prove Theorem \ref{s_1}, we define the bilinear maps $\phi_1 : M_m \times M_n \to M_m$ and $\phi_2 : M_m \times M_n \to M_n$ by
$$
\phi_1(A,B) = ({\rm id} \otimes \tr)(\Theta(A \otimes B)), \qquad
\phi_2(A,B) = (\tr \otimes {\rm id})(\Theta(A \otimes B)).
$$
Since $\Theta$ is a linear isomorphism on $M_m \otimes M_n$ which maps $\mathcal S_1$ onto itself, $\Theta$ sends an extreme ray of $\cals_1$ onto
another extreme ray. Recall that every extreme ray of $\cals_1$ is generated by
$P \otimes Q$ for $P \in \mathcal E_m$ and $Q \in \mathcal E_n$. Therefore, $\Theta(P\ot Q)$ can be written by
$$
\Theta(P \otimes Q) =\lambda P' \otimes Q',
$$
for one dimensional projections $P', Q'$ and $\lambda > 0$.
Then, we have
$$
\lambda = \tr(\Theta(P \otimes Q)).
$$
Since
$$
\phi_1(P,Q) = \lambda P', \qquad
\phi_2(P,Q) = \lambda Q',
$$
we have
\begin{equation}\label{rep}
\Theta(P \otimes Q) = {1 \over \tr(\Theta(P \otimes Q))} \phi_1(P,Q) \otimes \phi_2(P,Q)
\end{equation}
and
\begin{equation}\label{trace}
\tr (\phi_1(P,Q)) = \tr(\Theta(P \otimes Q)) = \tr (\phi_2(P,Q)).
\end{equation}
We fix $Q \in \mathcal E_n$. The linear map $\phi_1(\,\cdot\,,Q) : M_m \to M_m$ (respectively,
$\phi_2(\,\cdot\,,Q) : M_m \to M_n$) maps $\mathcal E_m$ to $\mathcal E_m$ (respectively, $\mathcal E_m$ to $\mathcal E_n$).
By Lemma \ref{rank1preserving}, $\phi_1(\,\cdot\,,Q)$ (respectively, $\phi_2(\,\cdot\,,Q)$) are one of the following forms;
\begin{enumerate}
\item[(A1)] $A \mapsto S^*AS$ for an invertible $S \in M_m$ (respectively, a surjective $S \in M_{m,n}$ with $m \le n$);
\item[(A2)] $A \mapsto S^*A^\ttt S$ for an invertible $S\in M_m$ (respectively, a surjective $S \in M_{m,n}$ with $m \le n$);
\item[(B)] $A \mapsto f(A) R$ for $R \in \mathcal E_m$ (respectively, $R \in \mathcal E_n$) and a faithful positive functional $f$ on $M_m$.
\end{enumerate}
Note that all $S,R$ and $f$ depend on the choice of $Q$. We first show that
$\phi_1(\,\cdot\,,Q)$ is either of the form (A) for every $Q$, or of the form (B) for every $Q$.

Assume that it is possible that $\phi_1(\,\cdot\,,Q)$ has two different representations (A) and (B) at different choices of  $Q \in \mathcal E_n$.
Since $\mathcal E_n$ is path connected, it holds that either $\phi_1(\,\cdot\,,Q_k)$ of type (A) converges to $\phi_1(\,\cdot\,,Q)$ of type (B)
for some $Q_k, Q \in \mathcal E_n$, or vice versa.
It is impossible that the rank one operators $f_k(I)R_k$ converge to the invertible $S^*IS$.
Suppose that $S_k^*AS_k$ converges to $f(A)R$ for all $A \in M_m$.
Let $S$ be the cluster point of $S_k$ in $M_m$.
Then, we have $S^*AS=f(A)R$ for all $A \in M_m$, which implies that $S$ is not invertible.
Take a projection $P$ onto a vector orthogonal to the range of $S$.
Then, we have $f(P)R=S^*PS=0$, which contradicts that $f$ is faithful.
The similar argument also hold for (A2).
Hence, $\phi_1(\,\cdot\,,Q)$ is either of the form (A) for every $Q$, or (B) for every $Q$.

Assume that both $\phi_1(\,\cdot\,,Q)$ and $\phi_2(\,\cdot\,,Q)$ are of the form (A1); say
$$
\phi_1(\,\cdot\,,Q) = S_1^* \,\cdot\, S_1, \qquad \phi_2(\,\cdot\,,Q) = S_2^* \,\cdot\, S_2.
$$
By (\ref{rep}), we have
$$
\Theta(P \otimes Q) = \lambda_P S^*(P \otimes P)S
$$
for $P \in \mathcal E_m$, $\lambda_P = 1/\tr(\Theta(P \otimes Q))$ and $S = S_1 \otimes S_2$.
From this, we also have
$$
\Theta((P_1+P_2) \otimes Q) = \Theta(P_1 \otimes Q) + \Theta(P_2 \otimes Q) = S^*(\lambda_1 P_1 \otimes P_1 + \lambda_2 P_2 \otimes P_2)S
$$
for $\lambda_i = 1/\tr(\Theta(P_i \otimes Q))$.
Since $S$ is surjective, the condition $P_1+P_2=P_3+P_4$ implies that
$$
\lambda_1 P_1 \otimes P_1 + \lambda_2 P_2 \otimes P_2 = \lambda_3 P_3 \otimes P_3 + \lambda_4 P_4 \otimes P_4
$$
for $P_i \in \mathcal E_m$.
However, this is not possible, as we see with the following example
$$
P_1 = E_{11}, \quad P_2 = E_{22}, \quad P_3 = {1 \over 2} (E_{11}+E_{12}+E_{21}+E_{22}), \quad P_4 = {1 \over 2} (E_{11}-E_{12}-E_{21}+E_{22}).
$$
Since $P_1,P_2,P_3$ and $P_4$ are symmetric, we conclude that both $\phi_1(\,\cdot\,,Q)$ and $\phi_2(\,\cdot\,,Q)$ cannot be of the form (A) at the same time.

Next, assume that both $\phi_1(\,\cdot\,,Q)$ and $\phi_2(\,\cdot\,,Q)$ are of the form (B);
$$
\phi_1(\,\cdot\,,Q) = f_1(\,\cdot\,)R_1, \qquad \phi_2(\,\cdot\,,Q) = f_2(\,\cdot\,)R_2.
$$
By (\ref{rep}), we have
$$
\Theta(P \otimes Q) = {f_1(P)f_2(P) \over \tr(\Theta(P \otimes Q))} R_1 \otimes R_2,
$$
for all $P \in \mathcal E_m$, which contracts that $\Theta$ is surjective.

Therefore, we obtained the following dichotomy;
\begin{enumerate}
\item[(i)] $\forall Q \in \mathcal E_n$, $\phi_1(\,\cdot\,,Q)$ is of the form (A) and $\phi_2(\,\cdot\,,Q)$ is of the form (B),
\item[(ii)] $\forall Q \in \mathcal E_n$, $\phi_1(\,\cdot\,,Q)$ is of the form (B) and $\phi_2(\,\cdot\,,Q)$ is of the form (A).
\end{enumerate}
Similarly, we also have the dichotomy;
\begin{enumerate}
\item[(iii)] $\forall P \in \mathcal E_m$, $\phi_1(P,\,\cdot\,)$ is of the form (B)  and $\phi_2(P,\,\cdot\,)$ is the form (A),
\item[(iv)] $\forall P \in \mathcal E_m$, $\phi_1(P,\,\cdot\,)$ is of the form (A) and $\phi_2(P,\,\cdot\,)$ is of the form (B).
\end{enumerate}

Assume that both (i) and (iv) hold.
Fix $P_0 \in \mathcal E_m$ and $Q_0 \in \mathcal E_n$ and let $\phi_2(\,\cdot\,,Q) = f_Q(\,\cdot\,)R_Q$ and $\phi_2(P_0,\,\cdot\,) = f_0(\,\cdot\,)R_0$.
We have
$$
\phi_2(P,Q) = f_Q(P)R_Q = {f_Q(P) \over f_Q(P_0)} \phi_2(P_0,Q) = {f_Q(P) \over f_Q(P_0)} f_0(Q) R_0 = {f_Q(P) \over f_Q(P_0)} {f_0(Q) \over f_0(Q_0)} \phi_2(P_0,Q_0).
$$
Thus, $\phi_2(P,Q)$ is the scalar multiple of $\phi_2(P_0,Q_0)$ for all $P \in \mathcal E_m$, $Q \in \mathcal E_n$,
which contradicts that $\Theta$ is surjective.
Similarly, it is impossible that (ii) and (iii) hold at the same time.
It remains to consider the following two cases;
\begin{itemize}
\item
both cases (i) and (iii) hold,
\item
both cases (ii) and (iv) hold.
\end{itemize}

Suppose that (i) and (iii) hold, in particular, $\phi_1(\,\cdot\,,Q)$ and $\phi_2(P,\,\cdot\,)$ are of the form (A1).
Fix $P_0 \in \mathcal E_m$ and $Q_0 \in \mathcal E_n$ and let
$$
\begin{aligned}
\phi_1(\,\cdot\,,Q_0) = S_1^* \,\cdot\, S_1, \quad
&\phi_1(P, \,\cdot\,) = f_P(\,\cdot\,)R_P, \\
\phi_2(\,\cdot\,,Q) = g_Q(\,\cdot\,)R_Q, \quad
&\phi_2(P_0, \,\cdot\,) = S_2^* \,\cdot\, S_2.
\end{aligned}
$$
Then, we have
\begin{equation}\label{phi1}
\phi_1(P,Q) = f_P(Q)R_P = {f_P(Q) \over f_P(Q_0)} \phi_1(P,Q_0) =  {f_P(Q) \over f_P(Q_0)} S_1^*PS_1 =: F_P(Q) S_1^*PS_1,
\end{equation}
and
\begin{equation}\label{phi2}
\phi_2(P,Q) = g_Q(P)R_Q = {g_Q(P) \over g_Q(P_0)} \phi_2(P_0,Q) =  {g_Q(P) \over g_Q(P_0)} S_2^*QS_2 =:  G_Q(P)S_2^*QS_2.
\end{equation}
From
$$
\lambda F_P(Q) S_1^*PS_1 = \lambda \phi_1(P,Q) = \phi_1(\lambda P,Q) = F_{\lambda P}(Q) S_1^*(\lambda P)S_1,
$$
we see that
\begin{equation}\label{radial}
F_{\lambda P} = F_P, \qquad \lambda>0.
\end{equation}
Since
$$
\phi_1(P_1+P_2,Q) = \phi_1(P_1,Q) + \phi_1(P_2,Q) = S_1^*(F_{P_1}(Q)P_1 + F_{P_2}(Q)P_2)S_1,
$$
we have the implication
\begin{equation}\label{consistent}
P_1+P_2=P_3+P_4 ~\Rightarrow~ \lambda_1 P_1 + \lambda_2 P_2 = \lambda_3 P_3 + \lambda_4 P_4,
\end{equation}
for $\lambda_i=F_{P_i}(Q)$.

We proceed to show that $Q\mapsto F_P(Q)$ does not depend on a choice of $P$.
For this purpose, take one dimensional projections $P_1, P_2$ whose ranges are orthogonal.
We write
$$
U^*P_1U = E_{11}, \qquad U^*P_2U = E_{22}
$$
for a suitable unitary $U$.
Take
$$
P_3 = {1 \over 2}U(E_{11}+E_{12}+E_{21}+E_{22})U^*, \qquad P_4 = {1 \over 2}U(E_{11}-E_{12}-E_{21}+E_{22})U^*,
$$
which satisfies $P_1+P_2=P_3+P_4$.
The $2 \times 2$ left upper corners of $U^*(\lambda_1 P_1 + \lambda_2 P_2)U$ and  $U^*(\lambda_3 P_3 + \lambda_4 P_4)U$ are
$$
\begin{pmatrix}
\lambda_1 & 0 \\ 0 & \lambda_2
\end{pmatrix}
\quad \text{and} \quad
{1 \over 2}
\begin{pmatrix}
\lambda_3 + \lambda_4 & \lambda_3 - \lambda_4 \\ \lambda_3 - \lambda_4 & \lambda_3 + \lambda_4
\end{pmatrix},
$$
respectively, which implies $\lambda_1=\lambda_2$ by (\ref{consistent}).
Therefore, we have
$$
F_{P_1} = F_{P_2},
$$
for any rank one projections $P_1, P_2$ whose ranges are orthogonal.
Next, we take arbitrary linearly independent $P_1, P_2 \in \mathcal E_m$.
By the spectral decomposition of the rank two positive operator $P_1+P_2$, we write
$$
P_1+P_2 = \sigma_3 P_3 + \sigma_4 P_4,
$$
for rank one projections $P_3,P_4$ and $\sigma_3, \sigma_4 >0$.
By (\ref{consistent}), we have
$$
\lambda_1 P_1 + \lambda_2 P_2 = \lambda_3 (\sigma_3 P_3) + \lambda_4 (\sigma_4 P_4),
$$
where $\lambda_3 = F_{\sigma_3 P_3}(Q) = F_{P_3}(Q)$ and $\lambda_4 = F_{\sigma_4 P_4}(Q) = F_{P_4}(Q)$ by (\ref{radial}).
Since $P_3$ and $P_4$ are projections whose ranges are orthogonal, we have $\lambda_3=\lambda_4=:\lambda$.
It implies that
$$
\lambda_1 P_1 + \lambda_2 P_2 = \lambda (\sigma_3 P_3 + \sigma_4 P_4) = \lambda P_1 + \lambda P_2,
$$
thus $\lambda_1 = \lambda_2$.
Therefore, $F_P$ does not depend on the choice of $P \in \mathcal E_m$.
Similarly, $G_Q$ does not depend on the choice of $Q \in \mathcal E_n$.

By (\ref{trace}) and (\ref{phi1}), (\ref{phi2}), we get
$$
F(Q) \tr (S_1^*PS_1)= G(P) \tr(S_2^*QS_2).
$$
By the separation of variables, we conclude that
$$
{\tr (S_1^*PS_1) \over G(P)} = {\tr(S_2^*QS_2) \over F(Q)}
$$
is a constant, which we denote by $\kappa$.
The formulas (\ref{phi1}) and (\ref{phi2}) turn out to be
$$
\phi_1(P,Q) = {1 \over \kappa} \tr(S_2^*QS_2) S_1^*PS_1 \quad \text{and} \quad \phi_2(P,Q) = {1 \over \kappa} \tr(S_1^*PS_1) S_2^*QS_2.
$$
By (\ref{rep}) and (\ref{trace}), we have
$$
\Theta(P \otimes Q) = {1 \over \kappa} S_1^*PS_1 \otimes S_2^*QS_2.
$$
By a suitable scalar multiple of $S_1$ or $S_2$, we may assume $\kappa=1$ without loss of generality,
and so, obtain the following representation
$$
\Theta(A \otimes B) = S_1^*AS_1 \otimes S_2^*BS_2.
$$
When one of $\phi_1(\,\cdot\,,Q)$ and $\phi_2(P,\,\cdot\,)$ are of the form (A2), we get the representations
$$
\begin{aligned}
\Theta(A \otimes B) &= S_1^*A^\ttt S_1 \otimes S_2^*BS_2,\\
\Theta(A \otimes B) &= S_1^*A S_1 \otimes S_2^*B^\ttt S_2,\\
\Theta(A \otimes B) &= S_1^*A^\ttt S_1 \otimes S_2^*B^\ttt S_2.
\end{aligned}
$$

Finally, we consider the case when both cases (ii) and (iv) hold.
Since  $\phi_1(P,\,\cdot\,) : M_n \to M_m$ and $\phi_2(\,\cdot\,,Q) : M_m \to M_n$ are of the form (A), we have $m=n$ by Lemma \ref{rank1preserving}.
Since $\Theta \circ \text{\sf fl}$ satisfies (i) and (iii), we get the representations
$$
\begin{aligned}
\Theta(A \otimes B) = S_1^*B S_1 \otimes S_2^*AS_2&, \quad \Theta(A \otimes B) = S_1^*B^\ttt S_1 \otimes S_2^*AS_2, \\
\Theta(A \otimes B) = S_1^*B S_1 \otimes S_2^*A^\ttt S_2&, \quad \Theta(A \otimes B) = S_1^*B^\ttt S_1 \otimes S_2^*A^\ttt S_2,
\end{aligned}
$$
with $m=n$. This completes the proof of Theorem \ref{s_1}.

Therefore, we conclude that the following are equivalent for a Hermiticity preserving linear isomorphism $\Theta : M_m\ot M_n\to M_m\ot M_n$:
\begin{itemize}
\item
$\Theta(\cals_1)=\cals_1$,
\item
$\Gamma^\Theta : \phi \mapsto \choi^\Theta_\phi$
retains the correspondence between $\superpos_1$ and $\cals_1$,
\item
$\Gamma^\Theta : \phi \mapsto \choi^\Theta_\phi$
retains the correspondence between $\mathbb P_1$ and $\blockpos_1$,
\item
$\lan\ ,\ \ran_\Theta$ retains the duality between $\mathbb P_1$ and $\cals_1$,
\item
$\lan\ ,\ \ran_\Theta$ retains the duality between $\superpos_1$ and $\blockpos_1$,
\item
$\Theta$ is one of (\ref{exam_tr}) for nonsingular $s\in M_m$ and $t\in M_n$, or their composition.
\end{itemize}

As for the case of fixed $k$ with $1<k<m\meet n$, we conjecture that a given linear isomorphism $\Theta$ on $M_m\ot M_n$
satisfies $\Theta(\cals_k)=\cals_k$ if and only if
$\Theta$ satisfies $\Theta(\cals_k\setminus\cals_{k-1})=\cals_k\setminus\cals_{k-1}$ if and only if
$\Theta$ is one of (\ref{exam}) for nonsingular $s\in M_m$ and $t\in M_n$, or their composition.
Because the map $\ad_s$ preserves the trace if and only if $s$ is a unitary, the validity of the conjecture would imply that
$\Theta$ preserves both $\cals_k$ (respectively $\cals_k\setminus\cals_{k-1}$) and trace if and only if
$\Theta$ is one of (\ref{exam}) for unitaries $s\in M_m$ and $t\in M_n$, or their composition.
We recall that an isomorphism $\Theta$ on $M_m\ot M_n$ satisfying  $\Theta(\cals_k)=\cals_k$
sends rank one positive matrix whose range vector in $\mathbb C^m\ot \mathbb C^n$ has Schmidt rank $\le k$ to a positive matrix of same kind.
Therefore, \cite[Corollary 4.2]{johnston2011} tells us that if an isomorphism $\Theta$ satisfying $\Theta(\cals_k)=\cals_k$ for a fixed $k$
with $1\le k<m\meet n$ is completely positive
then $\Theta$ is $\ad_s\ot\ad_t$ for nonsingular $s\in M_m, t\in M_n$ or the flip operator with $m=n$ or their composition.

%%%%%%%%%%%%%%%%%%%%%%%%%%%%%%%%%%%%%%%%%%%%%%%%%%%%%%%%%%%%%%%%%%%%%%%%%%%%%%%%%%%%%%%%%
%%%%%%%%%%%%%%%%%%%%%%%%%%%%%%%%%%%%%%%%%%%%%%%%%%%%%%%%%%%%%%%%%%%%%%%%%%%%%%%%%%%%%%%%%
%%%%%%%%%%%%%%%%%%%%%%%%%%%%%%%%%%%%%%%%%%%%%%%%%%%%%%%%%%%%%%%%%%%%%%%%%%%%%%%%%%%%%%%%%
%%%%%%%%%%%%%%%%%%%%%%%%%%%%%%%%%%%%%%%%%%%%%%%%%%%%%%%%%%%%%%%%%%%%%%%%%%%%%%%%%%%%%%%%%
%%%%%%%%%%%%%%%%%%%%%%%%%%%%%%%%%%%%%%%%%%%%%%%%%%%%%%%%%%%%%%%%%%%%%%%%%%%%%%%%%%%%%%%%%
%%%%%%%%%%%%%%%%%%%%%%%%%%%%%%%%%%%%%%%%%%%%%%%%%%%%%%%%%%%%%%%%%%%%%%%%%%%%%%%%%%%%%%%%%

\end{document}